\documentclass[sigconf]{aamas}

\usepackage{thm-restate}

\newcommand{\epsdeltaone}{\epsilon_1}
\newcommand{\bestparetopayoffi}{u^*_i}
\newcommand{\bestparetopayoffj}{u^*_j}
\newcommand{\bestparetopayoffone}{u^*_1}

\newcommand{\prior}{q}

\newcommand{\allprogrn}{\allprog^{\text{rn}}}
\newcommand{\allprogreneg}{\allprog^{\text{RN}}}

\newcommand{\selection}{sel}%{bar}
\newcommand{\renegb}{\mathbf{\selection}}

\usepackage{amsthm}

\theoremstyle{definition}

\newtheorem{definition}{Definition} % Create the "definition" environment
\newtheorem{assumption}[definition]{Assumption}
\newtheorem*{assum:nopunish}{Assumption 7}

\newcommand{\pointx}{0.8}%0.65}
\newcommand{\tritop}{2.4}%1.95} % 3 x pointx
\newcommand{\tribottom}{0.267}%0.22} % pointx / 3

\usepackage{defaults}

\usepackage{balance} % for balancing columns on the final page

\setcopyright{ifaamas}
\acmConference[AAMAS '25]{Proc.\@ of the 24th International Conference
on Autonomous Agents and Multiagent Systems (AAMAS 2025)}{May 19 -- 23, 2025}
{Detroit, Michigan, USA}{A.~El~Fallah~Seghrouchni, Y.~Vorobeychik, S.~Das, A.~Nowe (eds.)}
\copyrightyear{2025}
\acmYear{2025}
\acmDOI{}
\acmPrice{}
\acmISBN{}

\acmSubmissionID{120}

\title[AAMAS-2025 Formatting Instructions]{Safe Pareto Improvements
for \\
Expected Utility
Maximizers in Program Games}

\author{Anthony DiGiovanni}
\affiliation{
  \institution{Center on Long-Term Risk}
  \city{London}
  \country{United Kingdom}}
\email{anthony.digiovanni@longtermrisk.org}

\author{Jesse Clifton}
\affiliation{
  \institution{Center on Long-Term Risk}
  \city{London}
  \country{United Kingdom}}
\email{jesse.clifton@longtermrisk.org}

\author{Nicolas Macé}
\affiliation{
  \institution{Center on Long-Term Risk}
  \city{London}
  \country{United Kingdom}}
\email{nicolas.mace@longtermrisk.org}

\begin{abstract}
Agents in mixed-motive coordination problems such as Chicken may fail to coordinate on a Pareto-efficient outcome. 
Safe Pareto improvements (SPIs) were originally proposed 
to mitigate miscoordination in cases where 
players lack probabilistic beliefs as to how their 
agents will play a game;
agents are instructed to behave so as to guarantee a Pareto improvement on how they would play by default.
More generally, SPIs may be defined as transformations of strategy profiles such that all players are necessarily better off under the transformed profile.
In this work, we investigate
the extent to which SPIs can reduce downsides of miscoordination
between expected utility-maximizing agents. We 
consider games in which players submit computer programs
that can condition their decisions on each other's code, 
and use this property to construct SPIs using programs
capable of \textit{renegotiation}. 
We first
show that under mild conditions on players’ beliefs,
each player always prefers to use renegotiation.
Next, we show that
under similar assumptions,
each player always prefers
to be willing to renegotiate  
at least to the point at which 
they receive the lowest
payoff they can attain in any efficient
outcome.
Thus 
subjectively optimal play
guarantees players at least these payoffs,
without the need for coordination
on specific Pareto improvements.
\end{abstract}


\newcommand{\BibTeX}{\rm B\kern-.05em{\sc i\kern-.025em b}\kern-.08em\TeX}

\begin{document}

%%% The following commands remove the headers in your paper. For final 
%%% papers, these will be inserted during the pagination process.

\pagestyle{fancy}
\fancyhead{}

%%% The next command prints the information defined in the preamble.

\maketitle 

%%%%%%%%%%%%%%%%%%%%%%%%%%%%%%%%%%%%%%%%%%%%%%%%%%%%%%%%%%%%%%%%%%%%%%%%

\section{Introduction}

Artificially intelligent (AI) systems will increasingly advise
or make decisions on behalf of humans, including in interactions with
other agents. Thus there is a need for research on \textit{cooperative AI}
\citep{dafoe2020open,Conitzer2023Jun}: How can we design AI systems that are capable of 
interacting with other players in ways that lead to high social welfare? 
One way that AI systems assisting humans could fail 
to cooperate is by failing to coordinate 
on one of several Pareto-efficient equilibria. This
risk is especially large in \textit{bargaining problems}, 
where players have different preferences over 
Pareto-efficient equilibria (think of the game of Chicken).
These problems are particularly prone to miscoordination, where 
each player uses a strategy that is part of some Pareto-efficient equilibrium, but collectively
the players' strategies are not an equilibrium.
Bargaining problems are ubiquitous, including in high-stakes
negotiations over climate change,
nuclear proliferation, or military disputes, making
them a crucial area of study for cooperative AI.

\par We will explore
how the ability of AI systems to condition 
their decisions on each other's inner workings
could reduce downsides
of
miscoordination
% coordination failure
in
bargaining problems.
% One direction for cooperative AI is to take advantage of 
% computer-based agents' abilities to 
% % verify each other's formal properties
% make 
% their inner workings transparent
% and to make
% credible conditional commitments. 
The literature on \textit{program 
equilibrium} 
has shown how games played by computer programs 
that can read each other's source code admit more cooperative equilibria
in other challenges for cooperation such as the Prisoner's Dilemma
\citep{tennenholtz2004program,lavictoire2014program,oesterheld2019robust}.
\textit{Safe Pareto improvements 
(SPIs)} \citep{oesterheld2021safe} were proposed 
as a mitigation for inefficiencies in settings where
players have delegates play a game on their behalf, 
and have Knightian uncertainty (i.e., lack probabilistic beliefs
\citep{knight1921risk})
about how their delegates will play. 
Under an SPI, players 
%instruct their delegates to
% alter
change
their
default policies so as to guarantee Pareto improvement on the 
default outcome. For example, consider two parties~\A{} and~\B{} who would by default go to war over some territory.
They might instruct their delegates to,
instead, accept the outcome
of a lottery that allocates the territory
to~\A{} with the probability that~\A{}
would have won the war.

We will consider the extent to which SPIs 
can
% improve
mitigate
inefficiencies from
miscoordination when (i) players
do have probabilistic beliefs and
maximize subjective expected utility
and (ii) games are played by
computer programs that can condition on their counterparts' 
source code.
Our goal is to establish
guarantees against miscoordination
in the well-studied program game setting. Relaxations of standard assumptions in this setting
--- e.g., players can precisely read each other's programs' source code,
can syntactically verify if a program
follows some template \citep{tennenholtz2004program},
and participate in the program
game in the first place --- are left to future work.
While this is an idealized framework,
insights from studying program games could be
applied to more realistic interactions
between actors with some degree of conditional commitment ability.
For example,
countries engaging in climate negotiations
might write bills that specify when the country
would be bound to some policies conditional
on the terms of other countries' bills \citep{heitzig2023improvinginternationalclimatepolicy}. 
And, smart contracts implemented on a blockchain could
execute commitments to
transactions conditional on other
actors' contracts \citep{varian2010computer,sun2023cooperativeaidecentralizedcommitment}.

Our contributions are as
follows: 

\begin{enumerate}
    \item We construct SPIs in the program game setting
    using programs that \textit{renegotiate}. Such programs
    have a ``default'' program; check if their default
    played against their counterparts' defaults
    results in an inefficient outcome; and, if so, 
    call a renegotiation routine in an attempt to 
    Pareto-improve on the default outcome. 
    We examine when renegotiation would be
    used by players who optimize
    expected utility given their beliefs about 
    what programs
    their counterparts will use (i.e., in \textit{subjective equilibrium}
    \citep{subjective}).
    % \jc{Check final 
    % framing/result here} 
    Under
    mild assumptions
    on players' beliefs, we show that
    SPIs are
    always used in subjective equilibrium
    (Propositions \ref{prop:spi} and \ref{prop:simultaneous}).
    \item We show that due
    to the ability to renegotiate, 
    under mild assumptions on players' beliefs,
    players
    always
    weakly prefer 
    programs
    that guarantee them
    at least
    the lowest payoff they can obtain
    on the Pareto frontier
    % in any Pareto-efficient
    % outcome
    (Theorem~\ref{prop:yudpoint}). Following 
    % \jc{Check reference} 
    \citet{Rabin1991AMO}, we call this
    payoff profile the \textit{Pareto meet minimum (PMM)}.
    Thus we provide for this setting a (partial) solution to the 
    ``SPI selection problem'' identified by 
    \citet{oesterheld2021safe}
    (hereafter, ``OC''), i.e., the problem that
    players must coordinate among SPIs in order to 
    %successfully
    Pareto-improve on default outcomes.
    The intuition for this is:
    The PMM is the most efficient
    point such that,
    no matter how aggressively the players bargain, no one expects to risk getting a worse
    deal by being willing to 
    renegotiate to that point.
    % The PMM is the most efficient
    % point from which all players 
    % still have a chance to 
    % renegotiate
    % to their ideal outcome.
    % \adigi{to do: maybe change, this intuition is a bit indirect}
    We also show in the appendix that the PMM bound is tight:
    In 
    mixed-motive games, 
    it is always 
    possible to find subjective equilibria
    in which players fail to Pareto-improve on the 
    PMM, even using iterated renegotiation
    (Proposition~\ref{proposition:inefficiency}). 
    % And given some beliefs, it
    % will be subjectively optimal
    % not to agree to anything
    % less than the player's ideal outcome,
    % despite the risk 
    % of failing to agree on a
    % Pareto-improvement. 
\end{enumerate}

\section{Related Work} 
\paragraph{Program equilibrium and commitment games.}
We 
build on  
program games,  
where computer players condition their actions
on each other's source code. 
Prior work has shown that the ability of computer-based
agents to condition their decisions on their
counterparts' programs can
%open up
enable
more efficient equilibria
\citep{mcafee,Howard1988May,rubinstein,tennenholtz2004program,lavictoire2014program,critch2019parametric,oesterheld2019robust,DiGiovanni2022Apr}. For example, 
%Tennenholtz's
\citet{mcafee}'s
program
\textit{``If other player's code == my code: Cooperate; Else: Defect''}
% $\texttt{``If other player's code == my code:}$ 
% \\
% $\texttt{Cooperate; Else: Defect''}$
is a Nash equilibrium of the program game version of the
one-shot Prisoner's Dilemma in which both players cooperate. 
% The literature on commitment games, in which players
% submit commitment devices that condition on each other's 
% devices without circularity, has produced similar results
(See also the literature
on commitment games, e.g.,
\citet{kalai2010commitment,FORGES201364}.) 
However, this literature focuses on 
the Nash equilibria of program games, 
rather than studying failure to coordinate
on a Nash equilibrium as we do.
\paragraph{Coordination problems and equilibrium selection.}
There are large theoretical and empirical
literatures on how agents might coordinate
in complete information bargaining problems
(see
\citet{schuessler2019focal} and 
references therein). 
Most closely related to this 
paper is the literature 
on whether communication before playing 
a simultaneous-move game
%of complete
%information
can improve coordination
% Several models of how
% players might coordinate in a one-shot interaction have
% been proposed, 
% including team reasoning 
% \citep{sugden2003logic,bacharach1999interactive,janssen2006strategic}
% and
% cognitive hierarchy theory
% \citep{bardsley2010explaining}. 
% problem \citep{Harsanyi1988}.
\citep{farrell1987cheap,Rabin1991AMO,farrell1996cheap,crawford2017let,he2019power}. 
\citet{Rabin1991AMO} considers solution concepts for games with 
pre-play communication called \textit{negotiated equilibrium (NGE)} and
\textit{negotiated rationalizability (NGR)}, where NGE assumes that
players know their counterpart's strategies exactly (up to randomization).
Rabin shows that under NGE players
are guaranteed at least their PMM payoff 
in bargaining problems, whereas under NGR they are not. NGR is closer to the 
notion of subjective equilibrium used
in
%the present
our
paper, which allows 
players to
have possibly-inaccurate beliefs about what programs their counterparts
will use.
\citet{Santos2000Jul} shows results analogous to \citet{Rabin1991AMO}'s
under cheap talk with alternating (rather than simultaneous) announcements.
% Following a suggestion by 
% \citet{Yudkowsky2013Sep}, 
% \citet{Diffractor2022Sep} develops a cooperative bargaining
% solution that uses the PMM as the disagreement point. 
Finally, OC proposed safe Pareto improvements
for mitigating inefficiencies from coordination failures.
% The one exception we are aware of is 
% % \citet{oesterheld2021safe}
% OC, who
% proposed safe Pareto improvements
% for mitigating inefficiencies from coordination, and showed how 
% in their setting SPIs could be implemented via program equilibrium. 
We
discuss
% \citet{oesterheld2021safe}
OC
and its connections to the present work
at greater length
%below.
in Section~\ref{sec:sub:pspi}.
% \adigi{add \citet{Yudkowsky2013Sep} and \citet{Diffractor2022Sep} somewhere; also maybe Alignment Forum
% looks more respectable?}
% In these models, communication can guarantee that rational 
% players coordinate on a Pareto efficient outcome in pure
% coordination problems, but not in bargaining problems. 
% Finally, various techniques for improving coordination
% between independently-trained AI systems 
% in bargaining problems have been 
% proposed in the multi-agent learning literature. 
% Recently, \citet{stastny2021normative} proposed
% an algorithm parameterized by a set of acceptable
% bargaining solutions, which allows players to 
% % eventually
% coordinate with players whose set 
% of acceptable bargaining solutions overlaps with theirs.
% players might fail to coordinate on an efficient program
% profile even when they can communicate
% before submitting their programs,
% as studied in literature on \textit{pre-play communication}.
% For example,
% \citet{Rabin1991AMO}
% considers a model of
% complete information games preceded by
% a sequence of (non-binding) negotiation messages.
% Even for arbitrarily long
% negotiation sequences,
% in this model
% efficiency is not guaranteed in
% bargaining problems (but
% it is guaranteed in pure coordination problems).

\section{Miscoordination and Safe Pareto Improvements
in Program Games} 

In this section, we introduce
the program games framework and subjective
equilibrium, the solution concept that is
our focus in this paper. Then we review 
OC's
% \citet{oesterheld2021safe}'s
safe Pareto improvements, and show how they
can be constructed in our setting using 
renegotiation.
Section~\ref{sec:disc}
contains a table summarizing
the notation used in this section and Section~\ref{sec:sub:selection}. Throughout the paper, our formalism will be for games with two players, for ease of exposition.
See appendix for
full proofs of our results in
the more general ${n}$-player
formalism.
The extension to ${n}$ players doesn’t introduce qualitatively new challenges. 
Intuitively, since players submit programs independently of each other, we
can apply the same arguments to the profile of counterparts for a given player, as we did to the single counterpart
in the two-player case.
% which 
% as well as the class of 
% Pareto improvements
% that
% are the focus of this paper.

% \subsection{Setup: Miscoordination and Subjective Equilibrium}
\subsection{Setup: Program Games and Subjective Equilibrium}
\label{sub:proggame}

Two players $i=1,2$ will play a ``base game'' 
of complete information $\game = (\allact = \actspaceone \times \actspacetwo, (\payone, \paytwo))$.
Let $\actspacei$
be the set
of possible actions for player $i$,
%with $\allact = \actspaceone \times \actspacetwo$,
and let $\payi(\fullact)$
be player $i$'s
payoff in~$\game$
when the players
follow an action profile 
$\fullact = (\actone, \acttwo)$.
Write 
$\fullpay(\fullact) = (\payone(\fullact), \paytwo(\fullact))$, and refer to the 
set of payoff profiles attainable 
by some $\fullact$ in~$\allact$ 
as the \termemph{feasible set}. 
Throughout, we use the index
$j$ for the player $j \neq i$.
For payoff profiles~$\textbf{x}$ and~$\textbf{y}$,
write~$\textbf{x} \succeq \textbf{y}$ if $x_i \geq y_i$ for all~$i$, and
$\textbf{x} \succ \textbf{y}$ if $x_i > y_i$ for all~$i$.

A program game $\game(\allprog)$ is a 
game in which a strategy
is a program that maps
the 
profile
of other players' programs
to an action in~$\game$.\footnote{We restrict to deterministic
programs for ease of exposition; the extension 
to probabilistic
programs,
as in, e.g., \citet{kalai2010commitment}, is straightforward.}
This way, each player's program implements
a commitment
to an action
conditional on the others' programs.
% In this
% example
% and those that follow,
Assume the action sets of $\game$ are continuous; this is practically without loss of generality, because our program game setting can be
extended to a setting where players can use correlated randomization (see, e.g., \citet{kalai2010commitment}).
%(conditioned on $\randomization$).
Here, $\allprog = \progspaceone \times \progspacetwo$,
where $\progspacei$ is a 
% (countable)
set of
computable functions
from~$\progspacej$
%\times \randdomain$
to~$\actspacei$.
We assume that all programs in $\progspacei$
halt against all programs in $ \progspacej$,
for each~$i$, as is standard in program game literature (see, e.g., \citet{tennenholtz2004program, oesterheld2019robust, oesterheld2021safe}).
(Each~$\progspacei$ can be viewed
as player~$i$'s ``default'' program
set,
which we will extend in
% a later
% section 
Section~\ref{sec:sub:pspi}
with a set of programs that
have a special structure.)
% We will see that for \textit{any} default program
% set,
% players are individually incentivized to instead
% use a program with this new structure.

Player $i$'s program is $\progi \in \progspacei$.
For a program profile~$\fullprog = (\progone, \progtwo)$,
abusing notation,
let the action profile played in the base game
by players with a given program profile
be $\fullact(\fullprog) = (\progone(\progtwo), \progtwo(\progone))$.
After all programs are simultaneously 
submitted,
%(potentially as draws from
%mixed strategies over programs),
%the value $\rand$ is drawn
%from the randomization signal $\randomization$,
%and 
the induced
action profile 
$\fullact(\fullprog)$
%$(\progone(\progtwo, \rand), \progtwo(\progone, \rand))$
is played in~$\game$.
Thus the payoff for player~$i$ in $\game(\allprog)$
resulting
from the program profile 
$\fullprog$
is
$\progpayi(\fullprog) = \payi(\fullact(\fullprog))$.
% $\progpayi(\fullprog) := \payi((\progone(\progtwo), \progtwo(\progone)))$.
%Write $\game(\allprog, \fullprogpay)$
%to denote $\game(\allprog)$
%together with associated functions
%$\fullprogpay = (\progpayone, \progpaytwo)$
%that map a program profile to 
%players' payoffs.
%Let $\fullprogpay(\fullprog) := (\progpayone(\fullprog), \progpaytwo(\fullprog))$.
%$\progpayi(\fullprog, \rand) := \payi((\progone(\progtwo, \rand), \progtwo(\progone, \rand)))$.
%The expected payoff
%is $\progpayi(\fullprog) := \mathbb{E} \payi((\progone(\progtwo, \randomization), \progtwo(\progone, \randomization)))$.

% \paragraph{Subjective equilibrium and miscoordination.}
To capture the possibility of 
miscoordination,
we do not assume a Nash equilibrium 
is played.
Instead, each player~$i$
has beliefs as to what program $\progj$ the
other player
will use, distributed
according to a probability
distribution $\priorij$
(whose support may be a superset
of~$\progspacej$).\footnote{Allowing
for $\priorij$ to be
supported on a \textit{superset} of $\progspacej$ will 
be important when we consider
extensions of players'
program sets
with SPIs
in Section~\ref{sec:sub:pspi}.}
% $\priorij$ as to what 
% %by default
% program 
% $\progj$
% player $j$ will use, 
% is distributed according to
% some \textit{beliefs} 
% $\priorij$
% supported on 
% a superset of
% $\progspacej$ (written $\progj \sim \priorij$).
%(In general, $\priorij$ might be conditioned on
%the player's own program~$\progi$,
%as discussed in Section~\ref{sub:beliefs}.)
Then, a subjective equilibrium 
\citep{subjective}
is a profile of programs and 
beliefs such that each
player's
program maximizes expected 
utility with respect to their beliefs: 

\begin{definition}%[Subjective equilibrium]
Let
% $\actfulleq = (\acteqone, \acteqtwo)$
$\fulleq = (\eqprogone, \eqprogtwo)$
and $\fullprior = (\prioronetwo, \priortwoone)$
be profiles of
programs
% actions
and beliefs, respectively, in 
% $\game$.
$\game(\allprog)$.
We say
$(\fulleq, \fullprior)$
% $(\actfulleq, \fullprior)$
is a \textbf{subjective equilibrium} of
$\game(\allprog)$ %,
% $\game$
%written $(\fulleq, \fullprior) \in \subjeq(\game(\allprog), \fullprogpay)$,
if, for each~$i$,
%$\eqprogi \in \bestresp_{\progpayi}(\priorij,\eqprogj)$.
$$\eqprogi \in \argmax_{\progi \in \progspacei}
\mathbb{E}_{\progj \sim \priorij} \progpayi(\fullprog).$$
% $$\acteqi \in \argmax_{\acti \in \actspacei} \mathbb{E}_{\actj \sim \priorij} \payi(\fullact).$$
\end{definition}

Subjective equilibrium is, of course,
a weaker solution concept than Nash equilibrium
(or even rationalizable strategies \citep{Bernheim1984,Pearce1984}).
The results in this paper that follow
will be stronger than showing that a given
strategy is used in \textit{some} subjective equilibrium.
Instead, we will construct strategies such that,
for \textit{any} beliefs players might have under some assumptions,
and \textit{any} program profile they consider using,
our strategies are individually
(weakly) preferred by players over that program profile
---
and are thus used in a subjective equilibrium
associated with those beliefs. Therefore,
considering subjective equilibrium
will make our results stronger
than if we had assumed players'
beliefs satisfied a Nash equilibrium assumption.

The base games we are interested in are \textit{bargaining problems}, where players can miscoordinate in 
subjective equilibrium if they are sufficiently 
confident their counterparts will play favorably to them.
% in which, first, players have differing preferences over some pair of action profiles,
% and second,
% an inefficient outcome 
% (``miscoordination'')
% results
% when each player plays from their most preferred action profile.
% % We will say that players \textit{miscoordinate}
% if their programs evaluate to this inefficient
% action profile.
%\adigi{probably replace with shorter informal part}
%We will focus on base games $\game$ that are \textit{bargaining problems}. 
%Writing $\fullpay(\fullact) = (U_1(\fullact), U_2(\fullact))$,
%this means 
%that there are at least two action profiles $\fullact$, $\fullact'$
%such that $\fullpay(\fullact)$ and $\fullpay(\fullact')$ are Pareto-efficient; 
%player~1 strictly prefers $\fullact$ to $\fullact'$; 
%player 2 strictly prefers~$\fullact'$ to $\fullact$; and
%$\fullpay(a_1, a_2')$ is inefficient.
%for $i=1,2$.
%(By contrast, in pure coordination problems,
%there is exactly one Pareto-efficient payoff profile.)
%We will say that players
%\textit{miscoordinate} if their programs 
%$\fullprog$
%evaluate to a profile $(a_i, a_j')$.
% Miscoordination may occur in subjective equilibrium if
% each player~$i$ places sufficiently
% high probability on~$j$ playing
% from~$i$'s preferred action profile.
This is possible even when players are capable
of conditional commitments
as in program games:
\begin{exmp}
\label{ex:subjeq}
\textbf{(Miscoordination in subjective equilibrium)} 
%\begin{quote}
Suppose
two principals delegate to AI assistants
to negotiate on their behalf
over the time for a meeting.
% on their behalf,
% to choose a time slot for a meeting.
Call this the Scheduling Game
(Table~\ref{tab:meeting}).
The principals meet if and only if
the AIs agree on one of three possible time slots.
Each principal~$i$
most
prefers slot~$i$, 
but would rather meet at slot~3
than not at all.
%\texttt{``always Slot $i$''}.
Suppose each player~$i$
% places
% sufficiently high
% probability on~$j$ 
thinks~$j$
is sufficiently
likely to
use\footnote{Abusing notation, we write ``$p_i=\langle \text{pseudocode for } p_i 
\rangle$'' to describe programs $p_i$.}
$\progj^{C} = $ \textit{``Slot $i$ if other player's code == `always Slot $i$'; Else: Slot~$j$''}. 
% \texttt{``Slot $i$ if other player's code == `always Slot $i$'; Else: Slot~$j$''}, 
% so,~$\progi^{D}$ is the
% only program
% that can exploit~$\progj^{C}$.
Intuitively, this program ``demands'' the player's best outcome,
except against~$\progi^{D} =$ 
\textit{``always Slot $i$''}, which exploits this program.
Each player might believe the other is likely to use~$\progj^{C}$
because it can both exploit programs that yield to its demand
and avoid miscoordinating with~$\progi^{D}$.
Then it 
is subjectively optimal 
for each player to submit
$\progi^{D}$.
The pair of programs $(\progone^{D}, \progtwo^{D})$
played in a subjective equilibrium
under these beliefs
results in the maximally inefficient (Slot~1, Slot~2) outcome.
%\end{quote}
\end{exmp}

\begin{table}
    \centering
    \caption{Payoff matrix for the Scheduling Game}
    \begin{tabular}{c|c|c|c|}
          \multicolumn{1}{c}{} & \multicolumn{1}{c}{Slot 1} & \multicolumn{1}{c}{ Slot 2} & \multicolumn{1}{c}{ Slot 3}\\
     \cline{2-4}
      Slot 1 & $3,1$ & $0, 0$ & $0, 0$ \\
        \cline{2-4}
       Slot 2 & $0,0$ & $1, 3$ & $0, 0$ \\
        \cline{2-4}
       Slot 3 & $0,0$ & $0, 0$ & $1, 1$ \\
       \cline{2-4}
    \end{tabular}
    \label{tab:meeting}
\end{table}

%\subsection{SPIs and renegotiation}
\subsection{Constructing
Safe Pareto Improvements via Renegotiation}
\label{sec:sub:pspi}

Informally, safe Pareto improvements (SPIs) \citep{oesterheld2021safe} are transformations $\fullfs$ of strategy profiles --- in our 
case, program profiles $\fullprog$ --- 
such that, for any $\fullprog$, all players are at least
as well off under $\fullfs(\fullprog)$ as under~$\fullprog$. 
% Oesterheld 
% and Conitzer
OC
focus on transformations induced by 
\textit{payoff} transformations,
and they formally 
define SPIs accordingly.
However, they note
that probability-1 Pareto improvements on 
players' 
default strategies can be achieved with other kinds of instructions
besides having 
delegates play a game with transformed payoffs (see OC, pg.~14).
Thus in this paper 
we define SPIs to be general transformations of 
strategy profiles that guarantee Pareto improvement:

\begin{definition}
\label{def:progSPI}
For a program game $\game(\allprog)$,
let
% $\fullfs : \allact \to \allact$
$\fullfs : \allprog \to \allprog'$
be a function of program profiles,
% and decision procedures,
written
$\fullfs(\fullprog) = (\fsprogone(\progone), \fsprogtwo(\progtwo))$,
for some joint program space $\allprog' = \progspaceone' \times \progspacetwo'$.\footnote{The assumption above that programs halt 
against each
other extends to $\allprog'$.}
% $\fullfs(\fullprog, \fulldecproc) = (\fsprogone(\progone, \decprocone), \fsprogtwo(\progtwo, \decproctwo))$.
%where each $\fsprogi$ is one-to-one.
Then~$\fullfs$
is an \textbf{SPI for $\game(\allprog)$}
if,
for all
% subjective equilibria
program profiles
$\fullprog$, we have
% $\progpayi(\fullfs(\fullprog)) \geq \progpayi(\fullprog)$ for both~$i$;
$\fullprogpay(\fullfs(\fullprog)) \succeq \fullprogpay(\fullprog)$;
and for
%each
some
program profile
$\fullprog$,
there is some
$i'$ such that $\progpayip(\fullfs(\fullprog)) > \progpayip(\fullprog)$.
% For a program game $(\game(\allprog), \fullprogpay)$,
% let
% $\fullfs : \allprog \to \fsallprog{\fullfs}$
% be some transformation of programs
% for some joint program space $\fsallprog{\fullfs} = \progspaceone^{\fullfs} \times \progspacetwo^{\fullfs}$,
% written $\fullfs(\fullprog) = (\fsprogone(\progone), \fsprogtwo(\progtwo))$.\footnote{The assumption above that programs halt 
% against each
% other extends to $\fsallprog{\fullfs}$.}
% %where each $\fsprogi$ is one-to-one.
% Then~$\fullfs$
% is a \textbf{SPI for $(\game(\allprog), \fullprogpay)$}
% if,
% in the extended program game $\game(\allprog \cup \fsallprog{\fullfs})$,
% for all subjective equilibria $(\fulleq, \fullprior)$ we have
% $\progpayi(\fullfs(\fulleq)) \geq \progpayi(\fulleq)$ for both $i$;
% and for
% %each
% some
% subjective equilibrium
% $(\fulleq, \fullprior)$
% such that
% %$(\progpayone(\fulleq), \progpaytwo(\fulleq))$
% $\fullprogpay(\fulleq)$
% is inefficient,
% there is some
% $i'$ such that $\progpayip(\fullfs(\fulleq)) > \progpayip(\fulleq)$.
\end{definition}

A natural approach to constructing an SPI is to
construct programs
that,
when they are all used against each other, map 
the action profile returned by
default programs to a Pareto improvement
whenever the default programs would have otherwise miscoordinated (i.e.,
the action profile is inefficient).
We call this construction ``renegotiation,'' and call mappings of action
profiles to Pareto improvements
``renegotiation functions.''\footnote{Compare to  section ``Safe Pareto improvements under improved coordination'' in OC.}

\begin{definition}
% We say
Call 
$\smreneg : \allact \to \allact$
a \textbf{renegotiation function}
if:
\begin{enumerate}
    \item For every $\fullact$,
    $\fullpay(\smreneg(\fullact)) \succeq \fullpay(\fullact)$.
    % is inefficient and both~$i$,
    % $\payi(\smreneg(\fullact)) \geq \payi(\fullact)$;
    \item For some $\fullact$
    % such that $\fullpay(\fullact)$
    % is inefficient 
    and some~$i'$,
    $\payip(\smreneg(\fullact)) > \payip(\fullact)$.
    % \item For every $\fullact$
    % such that $\fullpay(\fullact)$
    % is efficient,
    % $\smreneg(\fullact) = \fullact$.
    % \item For every $\fullact$
    % such that $\fullpay(\fullact)$
    % is inefficient,
    % $\fullpay(\smreneg(\fullact)) \succeq \fullpay(\fullact)$.
    % % is inefficient and both~$i$,
    % % $\payi(\smreneg(\fullact)) \geq \payi(\fullact)$;
    % \item For some $\fullact$
    % such that $\fullpay(\fullact)$
    % is inefficient and some~$i'$,
    % $\payip(\smreneg(\fullact)) > \payip(\fullact)$.
    % \item For every $\fullact$
    % such that $\fullpay(\fullact)$
    % is efficient,
    % $\smreneg(\fullact) = \fullact$.
\end{enumerate}
And let~$\fullrenegspace$ be the set of all renegotiation functions for the
given game~$\game$.
\end{definition}

We jointly define
the spaces of 
%\textit{renegotiating}
\textbf{renegotiation programs} $\onefsspacei(\smrenegi)$
for $i=1,2$
as those programs with
the structure of Algorithm~\ref{alg:renegotiation}, for
some:
\begin{itemize}
    \item renegotiation
    function
    $\smrenegi$ and
    \item ``default program''
    $\defprogi \in \progspacei \ \backslash \ \onefsspacei(\smrenegi)$.
\end{itemize}
(Note that the definition of Algorithm~\ref{alg:renegotiation}
for a given player~$i$ references the sets of programs given by 
Algorithm~\ref{alg:renegotiation} for
the \textit{other} player~$j$,
so this definition is not circular.)
% $\defprogi \in \progspacei$ where 
% $\defprogi \notin \onefsspacei(\smrenegi)$.
For any program profile~$\fullprog \in  \onefsspaceone(\smrenegone) \times \onefsspacetwo(\smrenegtwo)$ and any renegotiation function~$\smreneg$, we write 
$\deffullprog = (\defprogone, \defprogtwo)$
and
$\smreneg(\fullact) = (\renone(\fullact), \rentwo(\fullact))$.
%Programs with this structure are designed to avoid miscoordination
%by renegotiating with any program that is also willing to
%renegotiate in this way.

Renegotiation programs work as follows:
Consider the ``default outcome,'' the action profile
given by all players' default programs if they all use renegotiation programs (line~\ref{line:defout}).
Against any program~$\progj$
such that
the players' renegotiation functions (if any) don't all return the same Pareto improvement
on the default outcome,
% some player~$j$'s program is not in
% $\onefsspacej(\smrenegj)$,
% for any $\smrenegj$ such that the players' renegotiation functions
% agree on a Pareto improvement,
$\progi \in \onefsspacei(\smrenegi)$ plays according to its
default program $\defprogi$ (lines~\ref{line:definner} and~\ref{line:defouter} in Algorithm~\ref{alg:renegotiation}).
Against a program profile
that is willing to renegotiate
to the same Pareto improvement,
however,~$\progi$ plays
its part of the Pareto-improved outcome (line~\ref{line:renego}).
% action returned
% by
% the renegotiation function
% applied
% to the action
% profile
% given by all players' defaults.
% $\progi$ checks if its default
% would achieve an efficient payoff with
% the default of $\progj$. If so, the default program is run
% on the other player's default.
% If the defaults are miscoordinated, however,~$\progi$
% plays according to a
% renegotiation
% strategy $\reneg$, which returns
% new actions for both renegotiation programs.

\begin{algorithm}
\caption{Renegotiation program $\progi \in \onefsspacei(\smrenegi)$, for some $\defprogi$}
\label{alg:renegotiation}
\begin{algorithmic}[1]
\Require Counterpart program $\progj$
\If{$\progj \in \onefsspacej(\smrenegj)$ for some $\smrenegj$ $\in \fullrenegspace$}  \Comment{Check that $\progj$ renegotiates}
    \State $\deffullact \leftarrow \fullact(\deffullprog)$
    \label{line:defout}
    \If{$\smrenegi(\deffullact) = \smrenegj(\deffullact)$}
        \State \Return $\renii(\deffullact)$
         \Comment{Play renegotiation action}
         \label{line:renego}
    \Else
         \State \Return 
         $\defacti$
         \Comment{Play default against others' defaults}
         \label{line:definner}
    \EndIf
\Else
     \State \Return $\defprogi(\progj)$  \Comment{Play default}
     \label{line:defouter}
\EndIf
\end{algorithmic}
\end{algorithm}

It is easy to see that any 
possible Pareto improvement
(i.e., any possible mapping provided by a
renegotiation function)
can be implemented as an SPI via renegotiation programs:

\begin{proposition}
\label{prop:spi}
Let $\smreneg$ be a renegotiation function. 
For $i=1,2$, define 
$f_i: P_i \rightarrow \onefsspacei(\smreneg)$ such that,
for each $\progi \in P_i$, $\fsprogi(\progi)$
is of the form given in Algorithm 
\ref{alg:renegotiation} with 
% default program 
% equal to 
$\defprog{\fsprogi(\progi)} = \progi$. Then, the function
$\fullfs: \fullprog \mapsto (\fsprogone(\progone), \fsprogtwo(\progtwo))$ is an SPI.
\end{proposition}
\begin{proof}
This follows immediately from 
the definitions of renegotiation
function, Algorithm 
\ref{alg:renegotiation}, and 
SPI.
\end{proof}

 \begin{exmp}
\label{ex:SPIdemo}
\textbf{(SPI using renegotiation)} 
%\begin{quote}
In Example \ref{ex:subjeq},
the players miscoordinated
in the Scheduling Game.
However, each player $i$ might reason that,
if they were to renegotiate with~$j$, a
renegotiation function
that is fair enough to both players that
they would both be willing to use it is: Map
each outcome where players 
choose different slots
to 
the symmetric
(Slot~3, Slot~3) outcome.
% (and, similarly,
% maps other inefficient outcomes 
% to a point that provides a symmetric improvement).
So, they could be better off using transformed versions
of their defaults
that renegotiate in this way.
%\end{quote}
\end{exmp}

% \adigi{maybe want to flag up front the
% SPI selection problem, say we'll get to that?}

%\section{Belief-conditioned programs} 
% \section{Implementability
% and Individual Rationality}
% \section{Incentives to Use SPIs}
% \subsection{Individual Rationality of Renegotiation}
\subsection{Incentives to Renegotiate}
\label{sec:implement}

When is the use of renegotiation 
% individually rational
guaranteed in subjective equilibrium
in program games?
SPIs by definition make all players
(weakly) better off
\textit{ex post},
but
it remains to show that
players
prefer to renegotiate \textit{ex ante}. Intuitively, one might worry
that players
will choose not to accept a Pareto improvement
in order to avoid losing bargaining power.

It is plausible
that, all else equal,
players prefer strategies
that 
admit more opportunities for
coordination.
So, suppose
players always prefer
a renegotiation program
over a non-renegotiation program
if their expected
utility is unchanged.
Then Proposition~\ref{prop:simultaneous}
shows that,
under a mild assumption on players'
beliefs,
each player 
always
prefers to
transform their default program
into \textit{some}
renegotiation program.
That is, each player~$i$ always prefers to
use a program
in $\onefsspacei = \bigcup_{\smreneg} \onefsspacei(\smreneg)$
(so their program profile is in $\allprogrn = \onefsspaceone \times \onefsspacetwo$).
%apply \textit{some} SPI
%to their default program.
% For any mapping from inefficient action profiles to Pareto improvements, 
% there is an
% SPI that yields the same Pareto improvements
% and that
% players prefer to use in subjective equilibrium.

For this result, we assume 
(Assumption~\ref{assum:rennopun})
the following holds 
for any program profile $\fullprog$
given by (a) a program used by some player~$i$ in subjective
equilibrium and (b) a program
in the support of player~$i$'s beliefs:\
If the programs in $\fullprog$ don't
renegotiate with each other, then,
a program should respond equivalently to
any renegotiation program as it
would respond
to that program's default.
% for any $\progj$
% supported on players' beliefs or
% used
% in subjective equilibrium:\
% % players believe that
% Suppose
% the
% full program profile doesn't 
% \dots
% that every renegotiation program in a profile $\progmi$ would respond the same way to $\progj$ as its default would respond, i.e.,
% not renegotiate with $\progj$.
% Then $\progj$ never responds differently to (a) $\progmj$
% than to (b) the profile of defaults~$\defprogmj$
% (or the profile $\progmj$ with just $\progi$ replaced by its default).
% % No player's program $\progi$ responds
% % differently to (a) a 
% % profile~$\progmi$
% % containing 
% % renegotiation
% % programs
% % than to (b)
% % the profile of defaults~$\defprogmi$
% % (or the profile with o),
% % % its SPI transformation~$(\fsprogj(\progj))_{j \neq i}$,
% % if every pair~$\progj$ and~$\defprogj$ 
% % would respond identically to~$\progi$, i.e.,
% % not renegotiate with~$\progi$.
% %That is,
% %if player $i$'s program
% %doesn't renegotiate with $j$'s program,
% %and therefore outputs a response to~$j$
% %identical to that of some default program,
% %then~$j$'s response to~$i$'s program
% %is the same as to the default program.
This is because it seems implausible
that players would
respond
differently to renegotiation programs
that do not respond differently to them
(in particular, 
``punish'' renegotiation), all else equal.
% Additionally, we assume players believe there
% is positive probability
% that their best non-renegotiating
% program
% would miscoordinate with their counterparts.
(All full proofs are
in the appendix.)

\begin{assumption}\label{assum:rennopun}
We say that players with beliefs~$\fullprior$
% \textbf{are (i) certain 
% that renegotiation won't be punished and (ii) not certain of coordination}
\textbf{are certain 
that renegotiation won't be punished}
if the following holds.
%For any program profile~$\fullprog$, define~$\defprogmji$ as~$\progmj$ 
  %  with~$\progi$ replaced by~$\defprogi$.
    Take any renegotiation function $\smreneg$ $\in \fullrenegspace$;
    any renegotiation program
    $\progi \in \onefsspacei(\smreneg)$;
    and any~$\progj$ in the support of~$\priorij$
    such that the programs in $\fullprog$ don't renegotiate with each other.
    (I.e.,
    there is no $\smrenegj$
    such that 
    $\progj \in \onefsspacej(\smrenegj)$
    where $\smrenegj(\fullact(\deffullprog)) = \smreneg(\fullact(\deffullprog))$.)
    % (I.e.,~$\progi$
    % is a renegotiation program but
    % doesn't renegotiate with~$\progmi$.)
    % or is played in a subjective equilibrium
    % of~$\game(\allprog)$.
    Then:
    \begin{enumerate}
        \item $\progj(\progi) = \progj(\defprogi)$.
        \item If~$\defprogi$ is used in subjective equilibrium with respect
        to~$\priorij$,
        and
        $\progj \in \onefsspacej$,
        we have~$\defprogi(\progj) = \defprogi(\defprogj)$.
    \end{enumerate}
    %\label{assum:pun}
%     \item Let $\notfsspacei$ be the set of $\progi$
%     not in~$\onefsspacei$ (i.e., non-renegotiating
%     programs). For each~$i$,
%     there are some~$(\smrenegj)_{j \neq i}$ and~$\progmi  \in \bigtimes_{j\neq i} \onefsspacej(\smrenegj)$ in the support
%     of~$\priorij$
%     where:\
%      For any 
%      % subjectively optimal non-renegotiation program
%      $\genpi \in \argmax_{\progi \in \notfsspacei} \mathbb{E}_{\progmi \sim \priorij} \progpayi(\fullprog)$,
%     letting~$\deffullact = \fullact((\genpi, \defprogmi))$,
%      we have that~$\smrenegj(\deffullact)$ are equal
%      for all~$j$ and~$\payi(\smreneg(\deffullact)) > \payi(\deffullact)$.
%     % For~$\eqprogi \notin \onefsspacei(\smrenegi)$
%     % for any~$\smrenegi$, in subjective equilibrium of
%     % $\game(\allprog)$,
%     % letting~$\deffullact = \fullact((\eqprogi, \defprogj))$
%     %  we have~$\payi(\reni(\deffullact)) > \payi(\deffullact)$.
%     \label{assum:better}
%     % and assume this property holds for
%     % any~$\progj$ that is
%     % played in a subjective equilibrium
%     % of~$\game(\allprog)$.
% \end{enumerate}
\end{assumption}

% \citet{oesterheld2021safe},
% OC
% developed SPIs in a setting
% where principals have
% non-probabilistic beliefs
% about how their representatives
% will play a game,
% not necessarily as an approach
% that expected utility-maximizing
% agents will always prefer.
% Further,
% we may worry that similar
% problems arise as in games of incomplete
% information,
% where it is well-known
% that \textit{ex post} efficient
% outcomes
% might not be achievable in equilibrium
% due to players' uncertainty
% about each other \citep{myerson1983efficient}.

\begin{restatable}{proposition}{spirational}
\label{prop:simultaneous}
Let $\game(\allprog)$
be any program game.
Let~$\fullprior$ be any belief profile
satisfying the assumption
that players
are certain 
that renegotiation won't be punished 
%and not certain of coordination 
(Assumption~\ref{assum:rennopun}).
And, for some arbitrary renegotiation function~$\smreneg$, for each $i$ and $\progi \in \progspacei$, let $\irfsprogi(\progi)$
be the program
of the form in Algorithm~\ref{alg:renegotiation} with 
% default program 
% equal to 
$\defprog{\irfsprogi(\progi)} = \progi$.
% Then,
% %with probability 1,
% in every
% subjective equilibrium~$(\fulleq, \fullprior)$
% of $\game(\allprog \cup \allprogrn)$,
% for all~$i$ we have $\eqprogi \in \onefsspacei$.
Then,
%with probability 1,
for every
subjective equilibrium~$(\fulleq, \fullprior)$
of $\game(\allprog \cup \allprogrn)$
where $\eqprogip \notin \onefsspaceip$ for some~$i'$,
there exists $\fullprog' \in \allprogrn$
such that:
\begin{enumerate}
    \item For all~$i$, $$\progi' = \begin{cases}
        \eqprogi,& \text{if $\eqprogi \in \onefsspacei(\smrenegi)$ for some $\smrenegi$;} \\
        \irfsprogi(\eqprogi),& \text{else.}
    \end{cases}$$
    \item $(\fullprog', \fullprior)$
    is a subjective equilibrium
    of  $\game(\allprog \cup \allprogrn)$.
\end{enumerate}
\end{restatable}

\begin{sketch}
For any non-renegotiation
program for player~$i$,
construct
a renegotiation program
by letting this program be
% Consider any non-renegotiation program for player~$i$,
% and let it be
the default of
Algorithm~\ref{alg:renegotiation}.
% for some renegotiation function.
% For any player~$j$ program,
% there
% are two cases:\ 
If the other players'
programs
don't renegotiate
to the same outcome
as~$i$'s program,
then~$i$ uses their default,
so by the no-punishment assumption
they achieve the same payoff as in the original subjective
equilibrium.
Otherwise,
renegotiation
Pareto-improves on the default,
so the player is 
% at least
% as well
better off
using the renegotiation program.
% (and strictly better off with positive probability).
\end{sketch}

\section{The SPI Selection Problem and Conditional Set-Valued Renegotiation}
\label{sec:sub:selection}

To Pareto-improve on the default outcome, 
the renegotiation
programs defined
% above
in Section~\ref{sec:sub:pspi}
require players to coordinate on the renegotiation
function.
So does renegotiation just reproduce the same coordination
problem it was intended to solve? 
% \citet{oesterheld2021safe}
This is a general problem for SPIs, 
referred to by OC
as the ``SPI selection problem.''\footnote{OC give a brief informal
characterization of an
idea similar to our proposed partial
solution to SPI selection (p.\ 39): ``To do so, a player
picks an instruction that is very compliant (“dove-ish”) w.r.t. what SPI is
chosen, e.g., one that simply goes with whatever SPI the other players demand as long as that SPI cannot further be safely Pareto-improved upon.''
However, our approach does not require
complying with whatever SPI the other player demands.}

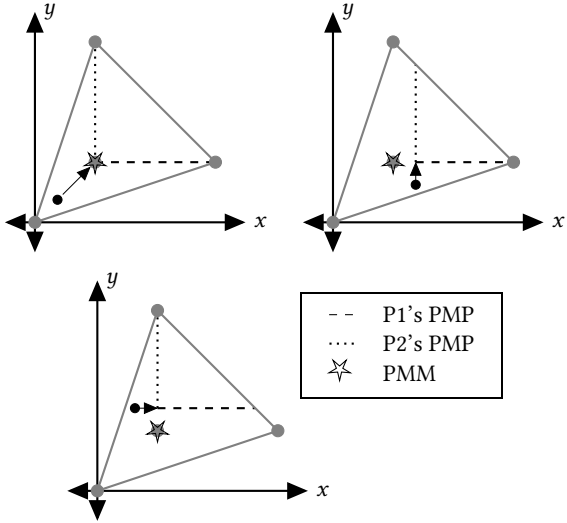
\begin{figure}
    \centering

    \begin{tabular}{cc}
         \begin{tikzpicture}[scale=1]
    \usetikzlibrary {shapes.geometric} 

% draw axes   
\draw[thick,<->] (-0.4, 0) -- (2.8, 0) node[right] {$x$};
\draw[thick,<->] (0, -0.4) -- (0, 2.8) node[right] {$y$};

\draw[thick,dashed,-] (0.8, 0.8) -- (2.4, 0.8);
\draw[thick,dotted,-] (0.8, 0.8) -- (0.8, 2.4);

%\fill[pattern=north east lines,pattern color=black] (0.8,0.75) -- (0.8,0.85) -- (2.4, 0.85) -- (2.4, 0.75) -- cycle;
%\fill[pattern=north west lines,pattern color=black] (0.75,0.8) -- (0.85,0.8) -- (0.85,2.4) -- (0.75, 2.4) -- cycle;

% draw points  
\draw[gray, fill=gray] (0, 0) circle (0.08);
\draw[gray, fill=gray] (2.4, 0.8) circle (0.08);
\draw[gray, fill=gray] (0.8, 2.4) circle (0.08);
\draw[gray, fill=gray] (0.8, 0.8) circle (0.08);
\node [star,
 star point height=.01cm,
 minimum size=0.01cm, 
 star point ratio=0.3,
 draw]
       at (\pointx,\pointx) {};

% \draw[gray,thin,dashed,-] (0.8, 0.8) -- (2.4, 0.8);
% \draw[gray,thin,dashed,-] (0.8, 0.8) -- (0.8, 2.4);

%\draw[black, fill=black] (0.35, 0.65) circle (0.06);
%\draw[black, fill=black] (0.65, 0.35) circle (0.06);
\draw[black, fill=black] (0.3, 0.3) circle (0.06);
%\draw[black, fill=black] (1.1, 0.5) circle (0.06);
%\draw[black, fill=black] (0.5, 1.1) circle (0.06);
% \draw[black, fill=black] (1.6, 0.6) circle (0.06);
% \draw[black, fill=black] (0.6, 1.6) circle (0.06);

 % \draw[black,thin,->,line width=0.02pt] (0.425, 0.675) -- (0.7, 0.77);
 %  \draw[black,thin,->,line width=0.02pt] (0.67, 0.44) -- (0.77, 0.7);
   \draw[black,thin,->,line width=0.02pt] (0.37, 0.37) -- (0.73, 0.73);
 %  \draw[black,thin,->,line width=0.02pt] (0.55, 1.1) -- (0.79, 1.1);
 %  \draw[black,thin,->,line width=0.02pt] (1.1, 0.55) -- (1.1, 0.79);
 %  \draw[black,thin,->,line width=0.02pt] (0.65, 1.6) -- (0.79, 1.6);
 %   \draw[black,thin,->,line width=0.02pt] (1.6, 0.65) -- (1.6, 0.79);

% \draw[black,pattern=north east lines,line width=3.5pt,rounded corners=50pt] (0.8,0.8) -- (2.4,0.8);
% \draw[black, fill=white] (0.8, 0.8) circle (0.06);
% \draw[black, fill=white] (2.4, 0.8) circle (0.06);
% \draw[black, fill=white] (0.8, 2.4) circle (0.06);

% draw triangle
\draw[gray, thick] (0,0) -- (2.4,0.8) -- (0.8,2.4) -- cycle;

%    \tikz\fill[pattern=north east lines,pattern color=gray] (0,0) rectangle (0.3,0.3); & P1's set \
% \tikz\fill[pattern=north west lines,pattern color=gray] (0,0) rectangle (0.3,0.3); & P2's set
  
\end{tikzpicture}

&

  \begin{tikzpicture}[scale=1]
    \usetikzlibrary {shapes.geometric} 

% draw axes   
\draw[thick,<->] (-0.4, 0) -- (2.8, 0) node[right] {$x$};
\draw[thick,<->] (0, -0.4) -- (0, 2.8) node[right] {$y$};

% draw points  
\draw[thick,dashed,-] (1.1, 0.8) -- (2.4, 0.8);
\draw[thick,dotted,-] (1.1, 0.8) -- (1.1, 2.1);
% \fill[pattern=north east lines,pattern color=black] (1.1,0.75) -- (1.1,0.85) -- (2.4, 0.85) -- (2.4, 0.75) -- cycle;
% \fill[pattern=north west lines,pattern color=black] (1.05,0.8) -- (1.15,0.8) -- (1.15,2.1) -- (1.05, 2.1) -- cycle;
\draw[gray, fill=gray] (0, 0) circle (0.08);
\draw[gray, fill=gray] (2.4, 0.8) circle (0.08);
\draw[gray, fill=gray] (0.8, 2.4) circle (0.08);
\draw[gray, fill=gray] (0.8, 0.8) circle (0.08);
\node [star,
 star point height=.01cm,
 minimum size=0.01cm, 
 star point ratio=0.3,
 draw]
       at (\pointx,\pointx) {};

% \draw[gray,thin,dashed,-] (0.8, 0.8) -- (2.4, 0.8);
% \draw[gray,thin,dashed,-] (0.8, 0.8) -- (0.8, 2.4);

%\draw[black, fill=black] (0.35, 0.65) circle (0.06);
%\draw[black, fill=black] (0.65, 0.35) circle (0.06);
%\draw[black, fill=black] (0.3, 0.3) circle (0.06);
\draw[black, fill=black] (1.1, 0.5) circle (0.06);
%\draw[black, fill=black] (0.5, 1.1) circle (0.06);
% \draw[black, fill=black] (1.6, 0.6) circle (0.06);
% \draw[black, fill=black] (0.6, 1.6) circle (0.06);

 % \draw[black,thin,->,line width=0.02pt] (0.425, 0.675) -- (0.7, 0.77);
 %  \draw[black,thin,->,line width=0.02pt] (0.67, 0.44) -- (0.77, 0.7);
%   \draw[black,thin,->,line width=0.02pt] (0.37, 0.37) -- (0.7, 0.7);
 %  \draw[black,thin,->,line width=0.02pt] (0.55, 1.1) -- (0.79, 1.1);
   \draw[black,thin,->,line width=0.02pt] (1.1, 0.55) -- (1.1, 0.79);
 %  \draw[black,thin,->,line width=0.02pt] (0.65, 1.6) -- (0.79, 1.6);
%    \draw[black,thin,->,line width=0.02pt] (1.6, 0.65) -- (1.6, 0.79);

% \draw[black,pattern=north east lines,line width=3.5pt,rounded corners=50pt] (0.8,0.8) -- (2.4,0.8);
% \draw[black, fill=white] (0.8, 0.8) circle (0.06);
% \draw[black, fill=white] (2.4, 0.8) circle (0.06);
% \draw[black, fill=white] (0.8, 2.4) circle (0.06);

% draw triangle
\draw[gray, thick] (0,0) -- (2.4,0.8) -- (0.8,2.4) -- cycle;

%    \tikz\fill[pattern=north east lines,pattern color=gray] (0,0) rectangle (0.3,0.3); & P1's set \
% \tikz\fill[pattern=north west lines,pattern color=gray] (0,0) rectangle (0.3,0.3); & P2's set
  
\end{tikzpicture}

\\
\end{tabular}

  \begin{tikzpicture}[scale=1]
    \usetikzlibrary {shapes.geometric} 

% draw axes   
\draw[thick,<->] (-0.4, 0) -- (2.8, 0) node[right] {$x$};
\draw[thick,<->] (0, -0.4) -- (0, 2.8) node[right] {$y$};

% draw points  
\draw[thick,dashed,-] (0.8, 1.1) -- (2.1, 1.1);
\draw[thick,dotted,-] (0.8, 1.1) -- (0.8, 2.4);
% \fill[pattern=north west lines,pattern color=black] (0.75,1.1) -- (0.85,1.1) -- (0.85, 2.4) -- (0.75,2.4) -- cycle;
% \fill[pattern=north east lines,pattern color=black] (0.8,1.05) -- (0.8,1.15) -- (2.1,1.15) -- (2.1,1.05) -- cycle;
\draw[gray, fill=gray] (0, 0) circle (0.08);
\draw[gray, fill=gray] (2.4, 0.8) circle (0.08);
\draw[gray, fill=gray] (0.8, 2.4) circle (0.08);
\draw[gray, fill=gray] (0.8, 0.8) circle (0.08);
\node [star,
 star point height=.01cm,
 minimum size=0.01cm, 
 star point ratio=0.3,
 draw]
       at (\pointx,\pointx) {};

% \draw[gray,thin,dashed,-] (0.8, 0.8) -- (2.4, 0.8);
% \draw[gray,thin,dashed,-] (0.8, 0.8) -- (0.8, 2.4);

%\draw[black, fill=black] (0.35, 0.65) circle (0.06);
%\draw[black, fill=black] (0.65, 0.35) circle (0.06);
%\draw[black, fill=black] (0.3, 0.3) circle (0.06);
%\draw[black, fill=black] (1.1, 0.5) circle (0.06);
\draw[black, fill=black] (0.5, 1.1) circle (0.06);
% \draw[black, fill=black] (1.6, 0.6) circle (0.06);
% \draw[black, fill=black] (0.6, 1.6) circle (0.06);

 % \draw[black,thin,->,line width=0.02pt] (0.425, 0.675) -- (0.7, 0.77);
 %  \draw[black,thin,->,line width=0.02pt] (0.67, 0.44) -- (0.77, 0.7);
%   \draw[black,thin,->,line width=0.02pt] (0.37, 0.37) -- (0.7, 0.7);
   \draw[black,thin,->,line width=0.02pt] (0.55, 1.1) -- (0.79, 1.1);
%   \draw[black,thin,->,line width=0.02pt] (1.1, 0.55) -- (1.1, 0.79);
 %  \draw[black,thin,->,line width=0.02pt] (0.65, 1.6) -- (0.79, 1.6);
%    \draw[black,thin,->,line width=0.02pt] (1.6, 0.65) -- (1.6, 0.79);

% \draw[black,pattern=north east lines,line width=3.5pt,rounded corners=50pt] (0.8,0.8) -- (2.4,0.8);
% \draw[black, fill=white] (0.8, 0.8) circle (0.06);
% \draw[black, fill=white] (2.4, 0.8) circle (0.06);
% \draw[black, fill=white] (0.8, 2.4) circle (0.06);

% draw triangle
\draw[gray, thick] (0,0) -- (2.4,0.8) -- (0.8,2.4) -- cycle;

%    \tikz\fill[pattern=north east lines,pattern color=gray] (0,0) rectangle (0.3,0.3); & P1's set \
% \tikz\fill[pattern=north west lines,pattern color=gray] (0,0) rectangle (0.3,0.3); & P2's set

\node[anchor=north east] at (5.5,2.75) {
\fbox{
\begin{tabular}{ll}
    \tikz[baseline]{\draw[dashed] (0,0.1) -- (0.35,0.1);} & P1's PMP \\
 \tikz[baseline]{\draw[dotted,thick] (0,0.1) -- (0.35,0.1);} & P2's PMP \\
 \tikz \node [star,
 star point height=.01cm,
 minimum size=0.01cm, 
 star point ratio=0.3,
 draw]
       at (\pointx,\pointx) {}; & PMM \\
% $\diamond$ & \scriptsize{Points in P2's set} \\
% \tiny{$\square$} & \scriptsize{Points in P2's set}
\end{tabular}
}
};
  
\end{tikzpicture}

    \caption{Illustration of the Pareto meet projection (PMP) of three different
    outcomes (black points) in the Scheduling
    Game, for each player.
    Gray points
    represent payoffs at 
    each pure strategy profile.
    Each black point is mapped via a player's PMP (black arrows) to
    a set containing
    a)
    the ``nearest'' point in the Pareto meet
    and b)
    all points better for the given player and
    no better for the other player than (a).
    }
    \label{fig:projection}
    \Description{This figure shows three different payoff plots in the Scheduling Game, each illustrating how the Pareto meet projection (PMP) maps different outcomes to sets of points. The graph shows payoff space with Player 1's payoff on the x-axis and Player 2's payoff on the y-axis. Gray points mark the payoffs from pure strategy profiles: (3,1), (1,3), and (1,1). The three black points represent different outcomes being mapped via PMP: (0,0), (1,1), and (2,2). Black arrows show how each point is mapped to its corresponding PMP set. For example, the miscoordination outcome (0,0) is mapped to a set that includes both the PMM point (1,1) and points that are better for the projecting player but no better for the other player.}
\end{figure}

Here, we argue
that, although in part the players' initial bargaining
problem recurs in SPI selection,
players will always renegotiate
so that each attains at least
the worst payoff they can get in any
efficient outcome. Following
\citet{Rabin1991AMO} we call the profile of
these payoffs the
\termemph{Pareto meet minimum (PMM)}.
% For each outcome
% that is worse for at
% least one player than their
% PMM,
Player~$i$'s \termemph{Pareto meet projection (PMP)}
(Fig.~\ref{fig:projection})
maps each outcome
to the
set of
Pareto improvements
such that, first,
each player's payoff
is at least the PMM,
and second, the payoff
of $j \neq i$
is not increased
except up to the PMM.
% unless the original payoff was worse for
% them than their PMM.
% \dots
% outcomes
% where each such player's
% payoff is increased to the PMM,
% plus outcomes where player~$i$
% is better off 
We will prove our bound by 
arguing that if 
players attempt to negotiate 
a Pareto improvement on an outcome,
they always at least weakly 
prefer to be willing to 
negotiate to the PMP of that outcome.

\begin{definition}
Let~$\effic$ be the set of Pareto-efficient
action profiles in~$\game$.
% \adigi{maybe will define this relative
% to a set of outcomes}
Then the
\textbf{Pareto meet minimum (PMM)}
payoff profile
is $\pmm = (\min_{\fullact \in \effic} \payone(\fullact), \min_{\fullact \in \effic} \paytwo(\fullact))$.
% \adigi{wait i guess
% this is just the PM; not just boundary?
% this is the \textbf{relative}
% Pareto meet}
Player~$i$'s \textbf{Pareto meet projection (PMP)}
of an action profile~$\fullact$
is the set
$\yudproji(\fullact)$
of action profiles~$\genfullact$
such that
$\payi(\genfullact) \geq \max\{\pmmi, \payi(\fullact)\}$
and
$\payj(\genfullact) = \max\{\pmmj, \payj(\fullact)\}$.
% Let $\game(\allprog)$ be a program
% game.
% For each~$i$, let~$\fullact^{i}$ be an
% outcome that maximizes player~$i$'s payoff among Pareto-efficient outcomes.
% Suppose~$\fullact^{1} \neq \fullact^{2}$,
% and
% % Let $\game(\allprog)$ be a program
% % game
% % where there are at least two action profiles $\fullact$, $\fullact'$
% % % such that $\fullpay(\fullact)$ and $\fullpay(\fullact')$ are Pareto efficient; 
% % $\payone(\fullact^{1}) > \payone(\fullact')$
% % and
% % $\paytwo(\fullact') > \paytwo(\fullact)$;
% % and
% $\payone((\actone^{1}, \acttwo^{2})) < \payone(\fullact^{2})$ and
% $\paytwo((\actone^{1}, \acttwo^{2})) < \paytwo(\fullact^{1})$.
\end{definition}
% there are some Pareto improvements
% that players are 
% guaranteed to prefer
% to take
% % always 
% % incentivized to take
% no matter how they approach the remaining
% bargaining problem.  

We’ll start by giving an informal description of the 
algorithm we will use to prove the
% result
guarantee, called
\termemph{conditional set-valued renegotiation (CSR)}.
Next, we’ll describe different components of the algorithm in more depth.
Finally, we’ll formally present the algorithm
and the guarantee.
% achieve full
% Pareto efficiency.

\subsection{Overview of 
CSR}
%Iterated Renegotiation}
\label{sec:sub:iterated}
If players want to increase
their chances of Pareto-improving
via renegotiation,
without necessarily accepting renegotiation outcomes
that heavily favor their counterpart,
% but 
% do not agree on a particular mapping
% from the miscoordination outcome to the Pareto
% frontier,
they can report to each other \textit{multiple}
renegotiation
outcomes they each would find acceptable
and take Pareto improvements
on which they agree.
% decompose renegotiation
% into multiple rounds
% and take Pareto improvements
% in the rounds
% on which they agree.
CSR implements such an approach. 
% \textit{iterated renegotiation}.
Like Algorithm~\ref{alg:renegotiation},
%iterated renegotiation
CSR
involves default programs,
and checks whether the default programs of a profile
of
%counterpart
%iterated renegotiation
CSR
algorithms result in an efficient outcome.
If not, CSR moves to a renegotiation procedure that works as follows:
\begin{enumerate}
    \item \textbf{Renegotiation
    using conditional sets.} 
    At this stage, programs “announce” sets of points that Pareto-improve on the default and that they are willing to renegotiate to,
    conditional on the other player's program (see shaded regions in Fig.~\ref{fig:pmm}).
    If these sets overlap, the procedure
    continues to the second step; 
    otherwise the players revert to their 
    defaults.
    % current
    % defaults.
    
    % The use of \textit{sets} at this stage is crucial to the result
    % that players are guaranteed
    % their PMM payoff.
    % % at least their PMM payoff.
    % This is because 
    \ \ \ \ \  The intuition for using
    \textit{sets} at
    this stage is that (we will argue)
    % if a player’s 
    % program 
    % would
    % renegotiate to some
    % point worse
    % for them than the PMM,
    % % than this lower bound,
    % % non-PMM guaranteeing point,
    % they
    % can safely instead use
    this way a player can use
    a program that is willing to 
    renegotiate to
    a payoff 
    % above this lower bound
    above their PMM payoff,
    without risking
    miscoordination
    if the other player
    does not also choose
    this new payoff
    precisely (see Fig.~\ref{fig:pmm}).
    Renegotiation sets 
    that \textit{condition} on 
    the other player's renegotiation
    set function
    are
    crucial
    to the result
    that players are guaranteed
    their PMM payoff.
    This is 
    because unconditionally 
    adding an outcome to the
    renegotiation set might
    provide Pareto improvements against some possible counterpart program,
    but
    make the outcome 
    worse
    %than otherwise
    against some \textit{other} possible counterpart program
    (see Example~\ref{ex:cond}).
    % because if a player
    % is willing to renegotiate to some
    % outcome
    % that would provide a Pareto improvement
    % against some possible counterpart,
    % that outcome
    % might be worse for them
    % than
    % otherwise
    % against some \textit{other} possible
    % counterpart.
    % if players wish
    % to renegotiate
    % for multiple rounds (see step~3),
    % they need to avoid precluding
    % their preferred
    % outcomes
    % by accepting
    % Pareto
    % improvements
    % in earlier rounds.
    \item \textbf{Choosing a point in the agreement set.}
    Call the intersection of the sets players
    announced at the previous stage the “agreement set.''
    At this stage, a ``\selfunction{} function''
    chooses an outcome from the Pareto frontier
    of the agreement set,
    which the players play instead of their
    miscoordinated default outcome.
    (Section~\ref{sec:sub:components}
    discusses how players coordinate on the \selfunction{} function, without needing to solve a further bargaining
    problem.)
    % This outcome will not in general be
    % efficient among all feasible outcomes, 
    % however, because the agreement set might contain only inefficient outcomes.
    % If so, renegotiation proceeds to step~3. 
    % \item \textbf{Attempting to Pareto improve again.}
    % Call the point agreed to in the previous step the “partial agreement point.''
    % If they have not reached
    % their maximum number of renegotiation
    % rounds~${K}$,
    % players then repeat steps~1 and~2
    % using the partial agreement point
    % as the ``default.''
    % If they fail to agree in this next round of renegotiation,
    % then the partial agreement outcome obtains,
    % and players have at least achieved a Pareto improvement
    % on the default outcome.
\end{enumerate}

\subsection{Components of
%Iterated Renegotiation
CSR}
\label{sec:sub:components}

\subsubsection{Set-valued renegotiation}

To avoid the need
to coordinate on an exact renegotiation function,
players
can
use functions
that map miscoordinated outcomes
to sets of Pareto improvements
they each find acceptable.
(In examples, we'll abuse terminology by referring to action profiles
by their corresponding payoff profiles.)
Then, we suppose
the players follow some rule
(a \termemph{\selfunction{} function})
for choosing an efficient outcome from
their agreement set.
% This way (we will show), no matter which \selfunction{}
% function is applied,
% players can guarantee a Pareto improvement
% without coordination.

\begin{definition}
Let
${\mathbf{C}}$ $\hspace{-0.8mm} (\allact)$ be the set of
closed
subsets of $\allact$.\footnote{I.e., closed with respect to the topology on $\allact$
induced by the Euclidean distance $d(\fullact,\fullact') = ||\fullpay(\fullact) - \fullpay(\fullact')||$.}
% closed
% subsets
% of ${S}$.
Letting~$\svrenegspacei$
be a set of functions
from~$\svrenegspacej \times \allact$
to $\mathbf{C}(\allact)$,
a function 
% $\renegki \in \svrenegspacei$
$\renegi \in \svrenegspacei$
is a \textbf{set-valued renegotiation
function}
if, 
% for all~$\renegkj \in \svrenegspacej$:
for all~$\renegj \in \svrenegspacej$:
\begin{enumerate}
    \item 
    For all~$\fullact \in \allact$ and $\fullact' \in \renegi(\renegj, \fullact)$,
    we have
    $\fullpay(\fullact') \succeq \fullpay(\fullact)$. 
    % For any~$\fullact$,
    % if~${r} \in \renegki(\renegkj, \fullact)$
    % then
    % $\fullpay({r}) \succeq \fullpay(\fullact)$.
    \item
     For some~$\fullact$
    %such that~$\fullpay(\fullact)$
    %is inefficient
    and some~$\fullact' \in \renegi(\renegj, \fullact)$,
    we have
    $\payip(\fullact') > \payip(\fullact)$
    for some~$i'$.
    % \item 
    % If~$\fullpay(\fullact)$
    % is efficient,
    % $\renegi(\renegj, \fullact) = \{\fullact\}$.
\end{enumerate}

A function 
$\renegb$ $: {\mathbf{C}}$ $\hspace{-0.8mm} (\allact) \to \allact$
is a \textbf{\selfunction{} function}
if
$\renegb({S})$ is Pareto-efficient
among points in~${S}$.\footnote{Because each ${S} \in \mathbf{C}(\allact)$
is closed,
some points in~${S}$
are guaranteed to be Pareto-efficient among points in
${S}$.}
A \selfunction{} function is \textbf{transitive}
if, for all~${S}, {S'}$ such 
that~$\fullpay(\mathbf{x}) \succeq \fullpay(\renegb({S}))$
for all~$\textbf{x} \in {S'}$,
we have~$\fullpay(\renegb({S} \cup {S'})) \succeq \fullpay(\renegb({S}))$.
% A tuple $\reneg = (\renega, \renegb, \renegc)$
% is a \textbf{renegotiation strategy}
% if~$\renega$ is a set-valued renegotiation
% function,
% $\renegb$
% is a \selfunction{} function,
% and $\renegc$
% is a renegotiation function.
\end{definition}

One might worry that by assuming a fixed selection function,
we still haven't avoided the need for coordination. However, note that there is no \textit{bargaining} problem involved in coordinating on a selection function. To see this, consider two players who intended to use renegotiation programs with different selection functions. Each player could switch to using a program that used the other player's selection function, and modify their set-valued renegotiation function
so as to guarantee the same outcome as if the other player switched to \textit{their} selection function. (See Appendix~\ref{app:barg} for a formal argument.)
So the players are indifferent as to which selection function is used.
(Coordinating on a selection function is a \textit{pure} coordination problem, however; compare to the problem of coordinating on the programming language used in syntactic comparison-based program equilibrium \citep{tennenholtz2004program}.)
In the results that follow,
we will show that players can guarantee
the PMM no matter which
(transitive)
\selfunction{} function they use.

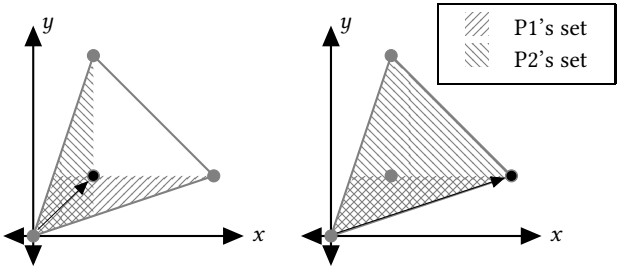
\begin{figure}
    \centering
    \begin{tabular}{cc}
    
\begin{tikzpicture}[scale=1]

% draw axes   
\draw[thick,<->] (-0.4, 0) -- (2.8, 0) node[right] {$x$};
\draw[thick,<->] (0, -0.4) -- (0, 2.8) node[right] {$y$};

% draw points  
\draw[gray, fill=gray] (0, 0) circle (0.08);
\draw[gray, fill=gray] (2.4, 0.8) circle (0.08);
\draw[gray, fill=gray] (0.8, 2.4) circle (0.08);
\draw[gray, fill=gray] (0.8, 0.8) circle (0.08);
% \draw[gray, fill=gray] (0.3, 0.8) circle (0.08);

% draw triangle
\draw[gray, thick] (0,0) -- (2.4,0.8) -- (0.8,2.4) -- cycle;

% shade below y=0.8
\fill[pattern=north east lines, pattern color=gray] 
  (0,0) -- (0.3,0.8) -- (2.4,0.8) -- cycle;
  
% shade left of x=0.8  
% \fill[pattern=north west lines, pattern color=gray]
%   (0,0) -- (0.8,0.3) -- (0.8,2.4) -- cycle;
% \fill[pattern=north west lines, pattern color=gray]
%   (0,0) -- (0.53,0.2) -- (0.53,1.6) -- cycle;
\fill[pattern=north west lines, pattern color=gray] 
(0,0) -- (\pointx, \tritop) -- (\pointx,\tribottom) -- cycle;

\draw[black, fill=black] (\pointx, 0.8) circle (0.06);

 \draw[black,thin,->,line width=0.02pt] (0.06, 0.085) -- (0.72, 0.72);
  
\end{tikzpicture}

&

\begin{tikzpicture}[scale=1]

% draw axes   
\draw[thick,<->] (-0.4, 0) -- (2.8, 0) node[right] {$x$};
\draw[thick,<->] (0, -0.4) -- (0, 2.8) node[right] {$y$};

% draw points  
\draw[gray, fill=gray] (0, 0) circle (0.08);
\draw[gray, fill=gray] (2.4, 0.8) circle (0.08);
\draw[gray, fill=gray] (0.8, 2.4) circle (0.08);
\draw[gray, fill=gray] (0.8, 0.8) circle (0.08);
% \draw[gray, fill=gray] (0.3, 0.8) circle (0.08);

% draw triangle
\draw[gray, thick] (0,0) -- (2.4,0.8) -- (0.8,2.4) -- cycle;

% shade below y=0.8
\fill[pattern=north east lines, pattern color=gray] 
  (0,0) -- (0.3,0.8) -- (2.4,0.8) -- cycle;
  
% shade left of x=0.8  
% \fill[pattern=north west lines, pattern color=gray]
%   (0,0) -- (0.8,0.3) -- (0.8,2.4) -- cycle;
% \fill[pattern=north west lines, pattern color=gray]
%   (0,0) -- (0.53,0.2) -- (0.53,1.6) -- cycle;
% \fill[pattern=north west lines, pattern color=gray] 
% (0,0) -- (0.4, 1.2) -- (0.4,0.15) -- cycle;
\fill[pattern=north west lines, pattern color=gray] 
(0,0) -- (0.8, 2.4) -- (2.4,0.8) -- cycle;

\draw[black, fill=black] (2.4, 0.8) circle (0.06);

 \draw[black,thin,->,line width=0.02pt] (0.08, 0.04) -- (2.3, 0.77);

\node[anchor=north east] at (3.9,3.2) {
\fbox{
\begin{tabular}{ll}
    \tikz\fill[pattern=north east lines,pattern color=gray] (0,0) rectangle (0.3,0.3); & P1's set \\
 \tikz\fill[pattern=north west lines,pattern color=gray] (0,0) rectangle (0.3,0.3); & P2's set \\
% $\diamond$ & \scriptsize{Points in P2's set} \\
% \tiny{$\square$} & \scriptsize{Points in P2's set}
\end{tabular}
}
};
  
\end{tikzpicture}

\\
    \end{tabular}
    \caption{
    Set-valued
    renegotiation in the Scheduling Game,
    for two possible player~2 renegotiation sets.
    % Gray points
    % represent payoffs at 
    % each pure strategy profile,
    % and 
    Black points represent renegotiation outcomes
    (mapped from the miscoordination
    outcome~$(0,0)$).
    If player~1 uses the renegotiation set
    shown here,
    they can achieve a Pareto improvement
    even if players don't reach the Pareto frontier (left),
    while still allowing for their
    best possible outcome (right).
    }
    \label{fig:pmm}
    \Description{This figure shows two plots illustrating set-valued renegotiation in the Scheduling Game, starting from a miscoordination outcome at (0,0). Both plots show Player 1's renegotiation set as a shaded region that includes their most preferred outcome (3,1) and all points that are Pareto-worse than (3,1). The left plot shows a case where Player 2's renegotiation set contains only their most preferred outcome (1,3), resulting in a renegotiation outcome (black point) that is better than (0,0) but not Pareto efficient. The right plot shows a case where Player 2's set also includes (3,1), allowing the players to achieve Player 1's most preferred outcome.}
\end{figure}

\begin{exmp}
\label{ex:svr}
\textbf{(Set-valued renegotiation)}
Suppose players in 
the Scheduling Game (Table~\ref{tab:meeting})
miscoordinate at $\fullact= (0,0)$.
The two plots in Fig.~\ref{fig:pmm} 
illustrate set-valued renegotiation
for
two possible player~2 renegotiation
sets
$\renegtwo(\renegone, \fullact)$,
and a fixed player~1
renegotiation
set
$\renegone(\renegtwo, \fullact)$.
%\footnote{We
    % In this
    % example
    % and those that follow,
%    assume players can use correlated
%    randomization
%    to play
%    any convex combination
%    of the base game
%    outcomes.}
Black points indicate the corresponding renegotiation
outcomes.
%\begin{quote}
Player~1 thinks it's likely that
the only efficient outcome
player~2 is willing to renegotiate to 
is their own most
preferred outcome (1,~3)
% (e.g., 
% player~2's set might be as depicted
% in the left 
(topmost gray point, left
plot).
But player~1 believes
that
with positive probability
player~2's renegotiation set 
% \adigi{make sure terminology consistent}
will also
include
player~1's most preferred outcome (3,~1)
(black point, right plot).
Player~1's best response given these
beliefs may be to choose a 
set-valued
renegotiation function~$\renegone$
that
maps (0,~0)
to a set
including both~(3,~1) and all outcomes
Pareto-worse
than~(3,~1),
i.e., the set depicted in Fig.~\ref{fig:pmm}.
This way, they
still achieve a Pareto improvement
if
player~2 has the smaller set (left plot),
and get their best payoff
if player~2 has the larger set (right plot).
%\end{quote}
\end{exmp}

\subsubsection{Conditional renegotiation sets.}
We saw that renegotiation sets allow a player to achieve 
Pareto improvements against a wider variety
of other players than is possible with renegotiation functions.
However, suppose a player could not condition their
renegotiation set on the other player's program.
Then, by adding a point to their renegotiation set
in attempt to Pareto-improve against some possible players,
they might lock themselves out of a better outcome against other possible players.

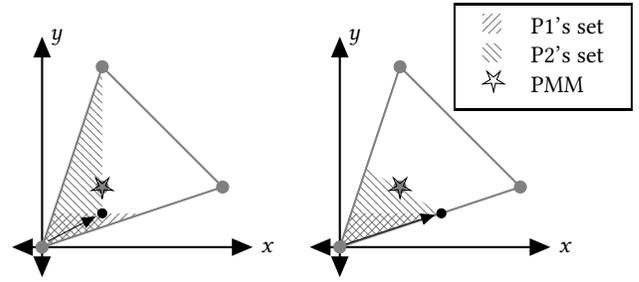
\begin{figure}
    \centering
    \begin{tabular}{cc}

    \begin{tikzpicture}[scale=1]
    \usetikzlibrary {shapes.geometric} 

% draw axes   
\draw[thick,<->] (-0.4, 0) -- (2.8, 0) node[right] {$x$};
\draw[thick,<->] (0, -0.4) -- (0, 2.8) node[right] {$y$};

% draw points  
\draw[gray, fill=gray] (0, 0) circle (0.08);
\draw[gray, fill=gray] (2.4, 0.8) circle (0.08);
\draw[gray, fill=gray] (0.8, 2.4) circle (0.08);
\draw[gray, fill=gray] (0.8, 0.8) circle (0.08);

% draw triangle
\draw[gray, thick] (0,0) -- (2.4,0.8) -- (0.8,2.4) -- cycle;

% shade below y=0.8
\fill[pattern=north east lines, pattern color=gray] 

  (0,0) -- (0.15, 0.45) -- (1.35,0.45) -- cycle;
  
% shade left of x=0.8  
\fill[pattern=north west lines, pattern color=gray]
    (0,0) -- (\pointx, \tritop) -- (\pointx,\tribottom) -- cycle;

    \draw[black, fill=black] (\pointx, 0.45) circle (0.06);

     \draw[black,thin,->,line width=0.02pt] (0.08, 0.08) -- (\pointx-0.1, 0.45-0.05);
   \node [star,
 star point height=.01cm,
 minimum size=0.01cm, 
 star point ratio=0.3,
 draw]
       at (\pointx,\pointx) {};
\end{tikzpicture}

    &

     \begin{tikzpicture}[scale=1]
    \usetikzlibrary {shapes.geometric} 

% draw axes   
\draw[thick,<->] (-0.4, 0) -- (2.8, 0) node[right] {$x$};
\draw[thick,<->] (0, -0.4) -- (0, 2.8) node[right] {$y$};

% draw points  
\draw[gray, fill=gray] (0, 0) circle (0.08);
\draw[gray, fill=gray] (2.4, 0.8) circle (0.08);
\draw[gray, fill=gray] (0.8, 2.4) circle (0.08);
\draw[gray, fill=gray] (0.8, 0.8) circle (0.08);

% draw triangle
\draw[gray, thick] (0,0) -- (2.4,0.8) -- (0.8,2.4) -- cycle;

% shade below y=0.8
\fill[pattern=north east lines, pattern color=gray] 

  (0,0) -- (0.15, 0.45) -- (1.35,0.45) -- cycle;
  
% shade left of x=0.8  
\fill[pattern=north west lines, pattern color=gray]
    % (0,0) -- (\pointx, \tritop) -- (\pointx,\tribottom) -- cycle;
    (0,0) -- (0.35, 1.05) -- (1.35, 0.45) -- cycle;

    \draw[black, fill=black] (1.35, 0.45) circle (0.06);

     \draw[black,thin,->,line width=0.02pt] (0.08, 0.04) -- (1.35-0.1, 0.45-0.025);
   \node [star,
 star point height=.01cm,
 minimum size=0.01cm, 
 star point ratio=0.3,
 draw]
       at (\pointx,\pointx) {};

\node[anchor=north east] at (4,3.35) {
\fbox{
\begin{tabular}{ll}
    \tikz\fill[pattern=north east lines,pattern color=gray] (0,0) rectangle (0.25,0.25); & P1's %~$\renegoneone$ 
    set \\
 \tikz\fill[pattern=north west lines,pattern color=gray] (0,0) rectangle (0.25,0.25); & P2's set \\
 \tikz \node [star,
 star point height=.01cm,
 minimum size=0.01cm, 
 star point ratio=0.3,
 draw]
       at (\pointx,\pointx) {}; & PMM \\
\end{tabular}
}
};

\end{tikzpicture}

\\
    \end{tabular}
    \caption{
   Two possible
    renegotiation procedures 
    %(black arrows)
    in the Scheduling Game,
    for different
    player~2 renegotiation sets.
    Player~1 might
    add the PMM (star) to their
    unconditional
    renegotiation set.
    In the case in the left plot, player~1 
    is no worse off by
    adding the PMM
    to their set.
    But in the case in the right plot,
    if player~1 adds
    the PMM,
    they might
    do worse
    if
    % player~2 bargains
    % for 
    the \selfunction{} function
    chooses
    the PMM instead
    of the
    black point
    that would have otherwise been
    achieved.
    }
    \label{fig:conditional}
    \Description{This figure shows two plots of renegotiation procedures in the Scheduling Game with different Player 2 renegotiation sets, illustrating why players need conditional renegotiation sets. A star marks the PMM outcome at (1,1). In the left plot, Player 1's unconditional renegotiation set includes both their original outcome (black point) and the PMM, and they are no worse off when the PMM is chosen. In the right plot, Player 1's unconditional set again includes both outcomes, but here they would do worse if the selection function chooses the PMM instead of the black point that would have been achieved with their original set.}
\end{figure}

\begin{exmp}
\label{ex:cond}
\textbf{(Failure of unconditional renegotiation sets)} 
%\begin{quote}
    Suppose that
    in the Scheduling Game,
    Player~1 uses an
    unconditional
    %first-round
    set-valued renegotiation function~$\renegone$.
    Fig.~\ref{fig:conditional}
    shows
    their set~$\renegone(\renegtwo, \fullact)$
    for the miscoordination
    outcome~$\fullact = (0,~0)$.
    % in the regions
    % shaded with down-sloped lines.
    Suppose player~1 
    %instead
    %uses~$\newrenegone$
    instead considers
    using~${\renegone}'$
    such that for all~$\renegtwo$,
    their renegotiation set is
    ${\renegone}'(\renegtwo, \fullact) = \renegone(\renegtwo, \fullact) \cup \{\pmm\}$. 
    For both player~2 renegotiation sets
    shown in the figure,
    the players renegotiate to 
    (i.e., the selection function chooses) the PMM (star).
    Then,
    player~1 is better off than under the default renegotiation outcome (black point)
    in the case in the left plot,
    but worse off
    in the case in the right plot. 
    % The PMM (star) Pareto-improves
    % on the default renegotiation outcome in the left plot (black point).
    % But suppose that, against a player~2 with the renegotiation set in the right
    % plot,
    % the players renegotiate to (i.e., the selection function chooses) the PMM.
    % Then, player~1 is worse off than under the renegotiation outcome that would have been reached (right black point) had they not included the PMM.
    % player~1 can no longer renegotiate
    % to~$(3,1)$
    % after they partially agree on~$\actfulleq$.
    But if player~1 had access
    to conditional
    renegotiation sets,
    they could instead use an~${\renegone}'$
    that
    includes the PMM
    \textit{only}
    against the player~2
    set in the left plot.
%\end{quote}
\end{exmp}

% \subsubsection{Pareto meet minimum. \adigi{title}}
\subsubsection{Renegotiation
sets that guarantee
the PMM}
How can a player guarantee
a payoff better
than some miscoordination outcome,
without losing the opportunity
to bargain for their most-preferred outcome?
% even when they are both 
% extremely confident
% the other will concede to their
% demand? 
Suppose
% the miscoordination outcome
% (e.g., any outcome with payoff~$(0, 0)$ in the Scheduling Game)
% is worse for each player than 
% % the payoff
% % they would receive from
% their
% least-preferred
% efficient outcome (i.e., their PMM).
player~$i$ considers
using a set-valued renegotiation
function that doesn't
guarantee the PMM.
That is,
against some counterpart 
program,
%profile,
the resulting outcome~$\fullact$ is worse for at least one player
than their 
least-preferred
efficient outcome (i.e., their PMM payoff). Then:
\begin{enumerate}
    \item As we will argue in Theorem~\ref{prop:yudpoint},
    under mild assumptions,
    player~$i$
    is no worse off also including 
    $\yudproji(\fullact)$ in their 
    renegotiation set. So, if player~$j$ also follows the same incentive to include $\yudprojj(\fullact)$ in their 
    renegotiation set, 
    these programs will guarantee at least the PMM.
    (Notice
    that players' ability
    to guarantee
    the PMM depends
    on conditional renegotiation sets, for
    the reasons discussed in
    Example~\ref{ex:cond}.)
    \item On the other hand, if~$i$ thinks
    the \selfunction{}
    function
    might not choose their optimal outcome in the agreement set,~$i$
    will \textit{not} prefer to include outcomes
    strictly better for player~$j$ than those in $\yudproji(\fullact)$.
    % And suppose
    % that 
    % player~$i$'s set-valued renegotiation
    % function
    % maps this miscoordination
    % outcome
    % at least to 
    % Pareto improvements
    % that leave~$j$ no better off than
    % under~$j$'s least-preferred efficient
    % outcome (i.e., player~$j$'s PMM).
    % \dots
    % for the first round,
    % player~$i$'s set-valued renegotiation
    % function
    % maps this miscoordination
    % outcome
    % only to Pareto improvements
    % that leave~$j$ no better off than
    % under~$j$'s least-preferred efficient
    % outcome (i.e., player~$j$'s PMM).
    For
    example,
    in
    Fig.~\ref{fig:pmm},
    $\renegone(\renegtwo, (0,0))$ includes all outcomes Pareto-worse than~$\yudprojone((0,0))$.
    This set
    safely
    guarantees
    the PMM against a player who
    also uses a set of this form,
    and gives player~1 their best possible outcome
    against the $\renegtwo$ in the right plot.
    But if  $\renegone(\renegtwo, (0,0))$ in the right plot included additional outcomes,
    which would be worse for player~1 than player~1's best possible outcome,
    %player~2 might bargain for
    the \selfunction{} function might choose
    an outcome
    that is worse for player~1
    than otherwise.
    (This is why, when we construct strategies
    for
    the proof of Theorem~\ref{prop:yudpoint},
    we add the entire PMP even though it is sufficient to only
    add
    the point in the PMP that minimizes the player's payoff.)
\end{enumerate}
% They can also safely include any outcomes strictly
% better for themselves than~$\fullact$ and no better
% for~$j$.

% lock out
% their most preferred outcome
% against the player~2 renegotiation set
% corresponding
% to the right plot.

%%%%%%%%%%%%%%%%%%%
%% COMBINED
%%%%%%%%%%%%%%%%%%%

% \subsection{Set-Valued Renegotiation
% Guarantees the PMM}
% \label{sec:sub:pmm}
% \subsection{Payoff Guarantees for
%Iterated Renegotiation}
% CSR}
% \label{sec:sub:pmm}

% \adigi{smth about focal points?}

% Call a tuple $\renegi = (\renegki)_{{k}=1}^{{K}} \in \bigtimes_{{k}=1}^{K} \svrenegspacei$
% a \textit{renegotiation
% strategy}.

Finally, here 
is the formal definition of CSR
programs. For a
set-valued
renegotiation
function~$\renegi$,
% renegotiation
% strategy~$\renegi$
% for a given number of rounds~${K}$,
we define
the space of
% iterated renegotiation
CSR
programs
$\fsspaceireni$
as the
space of programs
with the structure
of
Algorithm~\ref{alg:csr},
for some default~$\defprogi$.
(Let~$\fullsvspace$ be the set of all
set-valued renegotiation 
functions,
and for each~$i$
let the space of all
%iterated renegotiation
CSR
programs be~$\fullfsspacei = \bigcup_{\renegi \in \fullsvspace} \fsspaceireni$.
% let~$\fsspacei = \bigcup_{(\renegi, {K})} \fsspacei(\renegi, {K})$,
 Let $\allprogreneg = \fullfsspaceone \times \fullfsspacetwo$.)
The \selfunction{} function~$\renegb$
is given to
the players
(and we suppress dependence of~$\fullfsspacei$ on~$\renegb$ for simplicity).
% \adigi{unsure if we want to recap how it works}

\begin{algorithm}
\caption{
Conditional set-valued
renegotiation
% ${K}$-Iterated
% renegotiation
%Failsafe-implementing 
program $\progi \in \fsspaceireni$, for some $\defprogi$}
\label{alg:csr}
\begin{algorithmic}[1]
\Require Counterpart program $\progj$
\If{$\progj \in \fsspacej(\renegj)$ for some $\renegj \in \fullsvspace$}
    \State 
    $\deffullact \leftarrow \fullact(\deffullprog)$ 
        \State ${I} \leftarrow \renegone(\renegtwo, \deffullact) \cap \renegtwo(\renegone, \deffullact)$
    % \State $(\acceptone, \accepttwo) \leftarrow (\renegone(\renegtwo, \deffullact), \renegtwo(\renegone, \deffullact))$
    %     \State ${I} \leftarrow \acceptone \cap \accepttwo$
        \Comment{Agreement set}
         \If{${I} \neq \emptyset$}
    % and $\renegb({I}) = \renegjb({I}) $}
    % \Comment{Check agreement on first round}
    % \label{algline:firstagree}
             \State $\deffullact \leftarrow \renegb({I})$%\If{$(\progpayone((\defprog{\progone},\defprog{\progtwo})), \progpaytwo((\defprog{\progone},\defprog{\progtwo})))$ is efficient}
             \Comment{Renegotiation outcome}
    %     \State \Return $\defacti$  \Comment{Play default}
    %     \label{algline:default}
    \EndIf
    \State \Return $\defacti$
    \Comment{Play renegotiation outcome, or default}
\Else
     \State \Return $\defprogi(\progj)$  % \Comment{Play default}
\EndIf
\end{algorithmic}
\end{algorithm}

\subsection{Guaranteeing PMM Payoffs Using
%Iterated Renegotiation}
CSR}
\label{sec:sub:pmm}

% It is plausible
% that, all else equal,
% players prefer strategies
% that 
% admit more opportunities for
% coordination.
% So, 

Similar to the
assumption
in Section~\ref{sec:implement},
suppose
players always include
more outcomes in their renegotiation
sets if their expected
utility is unchanged.
So in particular, to show that in subjective equilibrium players use programs
that guarantee at least the PMM,
it will suffice to show that
they weakly prefer these programs.

Then, 
we will show in
Theorem~\ref{prop:yudpoint}
that
under
mild assumptions on players' beliefs,
for any program
that
does not guarantee
a player at least
their PMM payoff,
there is a corresponding CSR
program the player prefers
that \textit{does}
guarantee their PMM
payoff.
% and if players always include
% more outcomes in their renegotiation
% sets all else equal,
% then for any program
% played in subjective equilibrium
% that
% does not guarantee
% a player at least
% their PMM payoff,
% there is a corresponding CSR
% program played in subjective equilibrium
% that \textit{does}
% guarantee the PMM
% payoff.
We prove this
result
by constructing programs identical
to the programs
players would otherwise use,
except that these new programs'
renegotiation sets
for each outcome
include their PMP
of the
outcome they would have otherwise achieved.
For a program $\progi$, we call this 
modified program the 
\termemph{PMP-extension} of~$\progi$ 
(Definition \ref{def:pmp-extension}).
% if players expect
% each other to only
% use programs
% of the form of Algorithm~\ref{alg:csr}
% for one
% % round of renegotiation
% % (${K} = 1$),
% ${K} = 1$,
% they are guaranteed
% to achieve
% at least
% their PMM payoffs
% in \adigi{flag} subjective equilibrium.
% \adigi{flag}
% project each outcome~$\fullx$
% onto
% the PMP,
% %,''
% % i.e.,
% % the set of outcomes
% % where no more than one
% % player is better off than their PMM payoff,
% whenever~$\fullx$
% is Pareto-worse than this boundary.
% for
% every outcome
% where at least one player's
% payoff
% is worse than their
% PMM payoff,
% project
% that outcome
% onto the ``Pareto meet boundary,''
% i.e.,
% the set of outcomes
% where no more than one
% player is better off than their PMM payoff.
% \adigi{maybe figure, if we have space}

This result requires two
assumptions
on players' beliefs
and the structure of programs
used in subjective equilibrium
(Assumptions~\ref{assum:nopunish}\ref{subassum:csrwlog} and~\ref{assum:nopunish}\ref{subassum:extension}),
analogous to 
% Assumption~\ref{assum:rennopun}\ref{assum:pun}
Assumption~\ref{assum:rennopun}
of Proposition~\ref{prop:simultaneous}:
\begin{enumerate}
    \item Assumption~\ref{assum:nopunish}\ref{subassum:csrwlog}
is equivalent to
Assumption~\ref{assum:rennopun}
applied to CSR programs
rather than renegotiation programs:\ For any program
used in subjective equilibrium
or in the support of a player's
beliefs,
if that program never renegotiates,
it responds identically
to counterpart CSR programs
as to their defaults.
\item Informally, Assumption \ref{assum:nopunish}\ref{subassum:extension} 
says that players believe that, with probability~1:
 If a CSR program
is modified only
by adding PMP points
to its renegotiation set,
% the counterparts' renegotiation
% sets
% shouldn't respond differently
% if the new CSR program would respond
% identically to them.
% And if 
% the new CSR program would \textit{not}
% respond identically to the counterparts' renegotiation sets,
the only changes the counterparts
would
prefer to make are those that   
also add PMP points.
The intuition for this assumption is:\ For any possible default renegotiation outcome, the PMP-extension, by definition, doesn't add any points that make the counterpart strictly
better off than that outcome while making the focal player worse off (see Fig.~\ref{fig:pmp}).
So, similar
to Assumption~\ref{assum:nopunish}\ref{subassum:csrwlog}, the counterpart doesn't have an incentive
to make changes to their renegotiation set
that would make the focal player worse off.
(This argument wouldn't work if player~$i$ also added
outcomes that are \textit{better} for~$j$
than their PMP-extension. This is because, as noted in the previous section,~$j$ would then have an incentive to
exclude ~$i$'s most-preferred outcome from~$j$'s renegotiation set.)
\end{enumerate}
% were to make other changes to their renegotiation sets,
% the focal player would only be made worse off under the new renegotiation outcome
% if the counterparts were also made worse off. (This is because
% the PMP-extension, by definition, doesn't add any points that make counterparts strictly
% better off than some possible
% renegotiation outcome while making the focal player worse off.)

For Theorem~\ref{prop:yudpoint}
we also assume the \selfunction{} function is transitive.
This is an intuitive property:\ If outcomes are added to the agreement
set that make \textit{all} players weakly better off
than the default renegotiation outcome,
the new renegotiation outcome should be weakly better for all players.

\par The remainder of this subsection provides the formal 
details for the statement of Theorem \ref{prop:yudpoint}, 
and a sketch of the proof.

\multilinecomment{
\begin{figure}
    \centering

  \begin{tabular}{cc}

    \begin{tikzpicture}[scale=1]
    \usetikzlibrary {shapes.geometric} 

% draw axes   
\draw[thick,<->] (-0.4, 0) -- (2.8, 0) node[right] {$x$};
\draw[thick,<->] (0, -0.4) -- (0, 2.8) node[right] {$y$};

% draw points  
\draw[gray, fill=gray] (0, 0) circle (0.08);
\draw[gray, fill=gray] (2.4, 0.8) circle (0.08);
\draw[gray, fill=gray] (0.8, 2.4) circle (0.08);
\draw[gray, fill=gray] (0.8, 0.8) circle (0.08);

% draw triangle
\draw[gray, thick] (0,0) -- (2.4,0.8) -- (0.8,2.4) -- cycle;

\draw[thick,dashed,-] (0.8, 0.8) -- (2.4, 0.8);

% shade below y=0.8
\fill[pattern=north east lines, pattern color=gray] 

  (0,0) -- (0.15, 0.45) -- (1.35,0.45) -- cycle;
  
% shade left of x=0.8  
\fill[pattern=north west lines, pattern color=gray]
    (0,0) -- (\pointx, \tritop) -- (\pointx,\tribottom) -- cycle;

   \draw[black, fill=black] (\pointx, 0.8) circle (0.06);

 % \draw[black,thin,->,line width=0.02pt] (0.06, 0.085) -- (0.72, 0.72);
\end{tikzpicture}

    &

        \begin{tikzpicture}[scale=1]
    \usetikzlibrary {shapes.geometric} 

% draw axes   
\draw[thick,<->] (-0.4, 0) -- (2.8, 0) node[right] {$x$};
\draw[thick,<->] (0, -0.4) -- (0, 2.8) node[right] {$y$};

% draw points  
\draw[gray, fill=gray] (0, 0) circle (0.08);
\draw[gray, fill=gray] (2.4, 0.8) circle (0.08);
\draw[gray, fill=gray] (0.8, 2.4) circle (0.08);
\draw[gray, fill=gray] (0.8, 0.8) circle (0.08);

% draw triangle
\draw[gray, thick] (0,0) -- (2.4,0.8) -- (0.8,2.4) -- cycle;

\draw[thick,dashed,-] (0.8, 0.8) -- (2.4, 0.8);

% shade below y=0.8
\fill[pattern=north east lines, pattern color=gray] 

  (0,0) -- (0.15, 0.45) -- (1.35,0.45) -- cycle;
  
% shade left of x=0.8  
\fill[pattern=north west lines, pattern color=gray]
    (0,0) -- (\pointx-0.2, \tritop-0.55) -- (\pointx-0.2,\tribottom-0.08) -- cycle;

   \draw[black, fill=black] (\pointx-0.2, 0.45) circle (0.06);

 %\draw[black,thin,->,line width=0.02pt] (0.06, 0.085) -- (0.72, 0.72);

\node[anchor=north east] at (4,3.5) {
\fbox{
\begin{tabular}{ll}
    \tikz\fill[pattern=north east lines,pattern color=gray] (0,0) rectangle (0.25,0.25); & P1's %~$\renegoneone$ 
    set \\
 \tikz\fill[pattern=north west lines,pattern color=gray] (0,0) rectangle (0.25,0.25); & P2's set \\
 \tikz \node [star,
 star point height=.01cm,
 minimum size=0.01cm, 
 star point ratio=0.3,
 draw]
       at (\pointx,\pointx) {}; & PMM \\
\end{tabular}
}
};

\end{tikzpicture}

\\
    \end{tabular}
    
    \caption{Caption}
    \label{fig:enter-label}
\end{figure}
}

%%%%%%%%%%%%%%%%%%%%%%%
%%%%% Introducing assumptions outside of theorem 3 %%%%%

\begin{definition}\label{def:pmp-extension}
For any~$\progi \in \fsspaceireni$ for some $\renegi$, the \textbf{PMP-extension}~$\widetilde{\progi} \in \fsspaceinewreni$
    is the program
    identical to~$\progi$
    except:\
    for all $\progj \in \fsspacej(\renegj)$ for some $\renegj$,
    writing~$\newfullprogi = (\newprogi, \progj)$,
    we have
    \begin{align*}
        \newrenegi(\renegj, \fullact(\defnewfullprogi)) &= \renegi(\renegj, \fullact(\defnewfullprogi)) \cup \yudproji(\fullact(\fullprog)).
    \end{align*}
\end{definition}

\begin{assumption}
     We say that 
     players 
     with beliefs $\fullprior$ 
     \textbf{are (i) certain that CSR won't be punished and (ii) certain that PMP-extension won't be punished} 
     if the following hold:
     \begin{enumerate}[(i)]
     %     \item Suppose either $\defprogi$ is in a subjective equilibrium
     % of~$\game(\allprog \cup \allprogreneg)$,
     % or, for some player $j$, $\defprogi$ is an element of some $\progmj$ in the support of $\priorji$.
     % Let $\progi \in \fsspaceireni$ with $\renegi$ such that
     %     $\renegi({\renegj}', \fullact) = \emptyset$ for all ${\renegj}'
     %     , \fullact$,
     %     and default program $\defprogi$.
     %     Then 
     %     we have 
     %     $\defprogi(\defprog{\progj}) = \defprogi(\progj)$.
      \item Suppose either $\progi$ is in a subjective equilibrium
     of~$\game(\allprog \cup \allprogreneg)$,
     or $\progi$ is in the support of $\priorji$.
      Suppose $\progi \notin \fsspacei$. 
         Then for any $\progj \in  \fsspacej$,
         we have 
         $\progi(\progj) = \progi(\defprogj)$.
         \label{subassum:csrwlog}
         \item Let~$\progj \in \fsspacej(\renegj)$ be in the
     support of~$\priorij$,
     and take any~$\progi \in \fsspaceireni$
     with PMP-extension $\newprogi$.
         For all~$\fullact$, we have that
     $\renegj(\newrenegi,\fullact) = \renegj(\renegi,\fullact) \cup 
     {V}$ for some ${V} \subseteq \yudproji(\fullact(\fullprog))$. \label{subassum:extension}
     \end{enumerate}
     % (with corresponding PMP-extension~$\newprogi$).
     %we have
     % \begin{align*}
     %     \renegj(\renegmj,\fullact) \subseteq \renegj(\newrenegi,\fullact) \\
     %  \subseteq \renegj(\renegmj,\fullact) \cup 
     % \yudproji(\bargdefault(\renegi, \renegj, \fullact)) 
     % \end{align*}
     % $\renegj(\newrenegi,\fullact) = \renegj(\renegi,\fullact)$
     
     % and for~$\progj \notin \notsafe(\progi)$
     % we have
     % $$\renegj(\newrenegi,\fullact(\defnewfullprogi)) = \renegj(\renegmj,\fullact(\defnewfullprogi)).
     % $$
      \label{assum:nopunish}
\end{assumption}  % \adigi{unsure if this notation is intuitive}
%%%% End discussion of theorem 3 assumptions outside of theorem 3 %%%%%
% If modifying a
% % iterated renegotiation
% CSR
% program
% by adding PMP
% points
% to its renegotiation set
% doesn't change its behavior
% against some counterpart,
% that counterpart shouldn't
% respond differently.
% \adigi{not quite}
% %;
% and, \adigi{flag} there is nonzero probability
% that a player
% faces a counterpart against
% which they will
% do worse than their PMM.
%%%%%%%%%%%%%%%%%%%%%%%%
%%%% Option 1: stronger non-path-dependence
%%%%%%%%%%%%%%%%%%%%%%%%
\begin{figure}
    \centering
     \begin{tikzpicture}[scale=1]
    \usetikzlibrary {shapes.geometric} 

% draw axes   
\draw[thick,<->] (-0.4, 0) -- (2.8, 0) node[right] {$x$};
\draw[thick,<->] (0, -0.4) -- (0, 2.8) node[right] {$y$};

% draw points  
\draw[gray, fill=gray] (0, 0) circle (0.08);
\draw[gray, fill=gray] (2.4, 0.8) circle (0.08);
\draw[gray, fill=gray] (0.8, 2.4) circle (0.08);
\draw[gray, fill=gray] (0.8, 0.8) circle (0.08);
% \draw[gray,thin,dashed,->] (0.5, 1.18) -- (2.25, 0.8);
 % \node [star,
 % star point height=.01cm,
 % minimum size=0.01cm, 
 % star point ratio=0.3,
 % draw]
 %       at (\pointx,1.2) {};
% \draw[black, fill=black] (0.5, 1.2) circle (0.08);
% \draw[gray, fill=gray] (0.3, 0.8) circle (0.08);
% \draw[black,line width=5.5pt,rounded corners=50pt] (0.8,1.205) -- (2,1.205);
% draw triangle
\draw[gray, thick] (0,0) -- (2.4,0.8) -- (0.8,2.4) -- cycle;

% shade below y=0.8
\fill[pattern=north east lines, pattern color=gray] 
  % (0,0) -- (0.3,0.8) -- (2.4,0.8) -- cycle;
  % (0,0) -- (0.52, 1.6) -- (0.52, 0.6) -- (1.8,0.6) -- cycle;
(0,0) -- (0.4, 1.2) -- (0.4, 0.6) -- (1.8,0.6) -- cycle;
  
  % (0,0) -- (0.4,1.2) -- (2,1.2) -- (2.4, 0.8) -- cycle;
  
% shade left of x=0.8  
\fill[pattern=north west lines, pattern color=gray]
   % (0,0) -- (0.4, 1.2) -- (0.4,0.15) -- cycle;
   (0,0) -- (0.8, 2.4) -- (0.8,0.3) -- cycle;

\draw[black,fill=black] (-0.08, -0.08) rectangle (0.08, 0.08);

 \draw[black, fill=black] (0.4, 1.2) circle (0.06);

\draw[thick,dashed,-] (0.8, 1.2) -- (2, 1.2);

\node[anchor=north east] at (7,3) {
\fbox{
\begin{tabular}{ll}
 \tikz\node[rectangle,draw,fill=black] at (0.3,0.3) {}; & $\fullact(\deffullprog)$ \\
    \tikz\fill[pattern=north east lines,pattern color=gray] (0,0) rectangle (0.3,0.3); & $\renegone(\renegtwo, \fullact(\deffullprog))$ \\
 \tikz\fill[pattern=north west lines,pattern color=gray] (0,0) rectangle (0.3,0.3); & $\renegtwo(\renegone, \fullact(\deffullprog))$ \\
    \tikz\node[circle,draw,fill=black,scale=0.6] at (0,0) {}; & $\fullact(\fullprog)$ \\ 
 \tikz[baseline]{\draw[dashed] (0,0.1) -- (0.35,0.1);}
 & $\yudprojone(\fullact(\fullprog))$ \\
% $\diamond$ & \scriptsize{Points in P2's set} \\
% \tiny{$\square$} & \scriptsize{Points in P2's set}
\end{tabular}
}
};

%    \tikz\fill[pattern=north east lines,pattern color=gray] (0,0) rectangle (0.3,0.3); & P1's set \
% \tikz\fill[pattern=north west lines,pattern color=gray] (0,0) rectangle (0.3,0.3); & P2's set
  
\end{tikzpicture}
    \caption{Illustration of the argument for Theorem~\ref{prop:yudpoint}.
    By default, the renegotiation outcome is the black circle, $\fullact(\fullprog)$.
    Player~1
    considers whether to add
    % \dots
    % The best outcome for player~$i$ achievable
    % by further renegotiation
    % from this point is the triangle,~$\bargdefault(\renegtwoi, \renegtwoj, \bargdefault(\renegonei, \renegonej, \genfullact))$.
    % Player~$i$ is incentivized to add 
    to their renegotiation
    set $\renegone(\renegtwo, \fullact(\deffullprog))$
    the black striped segment $\yudprojone(\fullact(\fullprog))$.
    Player~1
    is certain that player~2
    would not change their set $\renegtwo(\renegone, \fullact(\deffullprog))$ 
    in response to this addition
    in a way that would make player~1
    worse off
    (Assumption~\ref{assum:nopunish}).
    This is because the only change player~1
    has made is to add outcomes that make both
    players weakly better off than~$\fullact(\fullprog)$
    and do not make player~2 strictly better off.
    % But player~$i$ does
    % not necessarily prefer to add outcomes above their PMP of the black circle,
    % because then they might do worse than their best
    % outcome in the PMP.
    }
    \label{fig:pmp}
    \Description{The default renegotiation outcome a(p) is marked with a black circle. A black striped segment shows Player 1's PMP of this outcome, PMP₁(a(p)), which includes points that make both players weakly better off than a(p) while not making Player 2 strictly better off. The shaded region represents Player 1's renegotiation set RN¹(RN², a(p^def)).}
\end{figure}
\begin{restatable}{theorem}{pmmguarantee}
\label{prop:yudpoint}
Let~$\game(\allprog)$ be a program game,
and~$\renegb$ be any transitive \selfunction{} function.
Suppose
the action sets of~$\game$
are continuous, so
that for any~$\fullact \in \allact$,
player~$i$'s PMP of that action profile
$\yudproji(\fullact)$
is nonempty.
Let~$\fullprior$ be any belief profile
satisfying the assumption 
that players are (i) certain that CSR won't be punished and (ii) certain that
PMP-extension won't be punished (Assumption~\ref{assum:nopunish}).

 Then,
 for any subjective equilibrium~$(\fullprog, \fullprior)$ of $\game(\allprog \cup \allprogreneg)$
 where $\progpayi(\fullprog) < \pmmi$ 
 for some~$i$,
 there exists~$\fullprog'$
 such that:
    \begin{enumerate}
        \item  For all~$i$, $\progi'$ is
     the PMP-extension of~$\progi$.
    \item $\fullprogpay(\fullprog') \succeq \pmm$.
    \item $(\fullprog', \fullprior)$
      is a subjective equilibrium
      of  $\game(\allprog \cup \allprogreneg)$.
    \end{enumerate}
 % \begin{enumerate}[(a)]
 %     \item 
 %     For all~$i$, $\progi'$ is
 %     the PMP-extension of~$\progi$.
 %      \item $\fullprogpay(\fullprog') \succeq \pmm$.
 %      % For both~$i$, $\progpayi(\fullprog') \geq \pmmi$.
 %      \item $(\fullprog', \fullprior)$
 %      is a subjective equilibrium
 %      of  $\game(\allprog \cup \allprogreneg)$.
 % \end{enumerate}
\end{restatable}
\begin{sketch}
% Consider the boundary formed by connecting
% each player's
% most-preferred efficient outcome
% to the [Yud point]
% (Figure \ref{fig:yud} \adigi{prob more descriptive?}).
Assumption~\ref{assum:nopunish}\ref{subassum:csrwlog}
implies that players always use CSR programs.
Consider
any renegotiation outcome~$\fullact$
worse for some player than the PMM,
which is achieved
by player~$i$'s ``old'' program
against some counterpart.
By Assumption~\ref{assum:nopunish}\ref{subassum:extension}, player~$j$ doesn't
punish~$i$ for adding
% the PMP of some outcome
their PMP
of that outcome,
$\yudproji(\fullact)$,
to their
%first-round
renegotiation set
(in their ``new'' program).
So the renegotiation outcome
of the new program
against~$j$
is only different
from that of the old
program
if~$j$ is also willing
to renegotiate
to some outcome in~$\yudproji(\fullact)$.
But in that case,
because the \selfunction{} function
is transitive,
the new renegotiation outcome
is no worse for~$i$
than under the old program.
Therefore,
each player always prefers to
replace a given program
with its PMP-extension,
and when all players
use PMP-extended programs,
the Pareto frontier of
their agreement set only includes
outcomes guaranteeing
each player
their PMM payoff.
\end{sketch}

\textit{Remark:} Notice that the argument above does not require that players refrain from using programs that implement other kinds of SPIs, besides PMP-extensions.
First, the PMP-extension can be constructed from any default program, including, e.g., a CSR program whose renegotiation set is only 
extended to include
the player’s best outcome,
not their PMP (call this a ``self-favoring extension'').
And if a player's final choice of program is their self-favoring extension, they are still incentivized to use the PMP-extension within their default program.

Second, while it is true that an analogous argument to the proof of Theorem~\ref{prop:yudpoint} could show that a player is weakly better off \textit{ex ante} using a self-favoring extension than
not extending their renegotiation set at all,
this does not undermine our argument.
This is because, as we claimed at the start of this section, it is reasonable to assume that among programs with equal expected utility, each player prefers to also include their PMP.
But wouldn't the player also prefer an even larger renegotiation set
that includes 
outcomes that Pareto-dominate the PMM as well?
No, because those outcomes will be worse for that player \textit{and}
better for their counterpart than the player’s most-preferred outcome,
such that the counterpart would have an incentive
to make the player worse off
(i.e., it's plausible
that Assumption~\ref{assum:nopunish}\ref{subassum:extension} would be violated).

\

We can now formalize the claim that
% iterated renegotiation
CSR
is an SPI
that partially solves
SPI selection:\ The mapping
from programs~$\fullprog$
to instances of 
Algorithm~\ref{alg:csr}
with$\fullprog$ as  
defaults,
for
\textit{any profile}~$(\renegone, \renegtwo)$
used in subjective equilibrium
under the assumptions of
% Proposition~\ref{prop:iterated},
Theorem~\ref{prop:yudpoint},
is an SPI that guarantees players
their PMM payoffs.

\begin{proposition}
\label{prop:itspi}
% Let $\renegone, \renegtwo$ be 
% renegotiation strategys.
For $i=1,2$,
for some \selfunction{} function~$\renegb$,
define 
$\fsprogi^{\renegi}: \progspacei \rightarrow \fsspaceireni$ such that,
for each $\progi \in \progspacei$, $\fsprogi^{\renegi}(\progi)$
is of the form given in Algorithm~\ref{alg:csr} with 
% default program 
% equal to 
$\defprog{\fsprogi^{\renegi}(\progi)} = \progi$.
Then, 
under the assumptions of 
%Proposition~\ref{prop:iterated},
Theorem~\ref{prop:yudpoint},
for any~$(\renegone, \renegtwo), \deffullprog$
such that
% \adigi{re-check why we have this condition}
for all~$\renegj$,
$\yudproji(\fullact(\fullprog)) \subseteq \renegi(\renegj, \fullact(\fullprog))$:
% $\progi \in \fsspacei(\renegi, {K})$
% of the form in Algorithm~\ref{alg:csr}
% are in a
% subjective equilibrium
% of $\game(\allprog \cup (\fullpspaceone \times \fullpspacetwo))$:
\begin{enumerate}
    \item The function
    $\fullfs^{\reneg}: \fullprog \mapsto (\fsprogone^{\renegone}(\progone), \fsprogtwo^{\renegtwo}(\progtwo))$ is an SPI.
    \item For all~$i$,
    $\progpayi(\fullfs^{\reneg}(\deffullprog)) \geq \max\{\progpayi(\deffullprog), \pmmi\}$.
\end{enumerate}
\end{proposition}

\begin{proof}
This follows
immediately
from the argument used to prove 
% Proposition~\ref{prop:iterated}.
Theorem~\ref{prop:yudpoint}.
\end{proof}

In
Appendix~\ref{app:proof:ineff},
we
show that 
players are \textit{not}
always incentivized
to use SPIs
that strictly
improve on the PMM. 

\begin{table}[ht]
	\caption{Key notation}
	\label{tab:notation}
	\begin{tabular}{rl}\toprule
		\textit{Symbol} & \textit{Description (page introduced)} \\ \midrule
         $\fullact(\fullprog)$ & action profile in the base game played by players \\
         & with the given program profile (2) \\
         $\smrenegi$ & renegotiation function for player~$i$ (maps an action \\
         & profile to a Pareto-improved action profile) (3) \\
		  $\renegi$ & set-valued renegotiation function for player~$i$  \\ & (maps $j$'s set-valued renegotiation function and  \\
    & an action profile to a set of Pareto-improved  \\
    & action profiles) (5) \\
    $\fullrenegspace, \fullsvspace$ & sets of all renegotiation functions and set-valued \\
    & renegotiation functions, respectively (3, 6) \\
        $\onefsspacei(\smrenegi)$ & set of renegotiation programs (Algorithm~\ref{alg:renegotiation}) that \\
        & use the renegotiation function~$\smrenegi$ (3)\\
        $\fsspaceireni$ & set of conditional set-valued renegotiation \\
        & programs (Algorithm~\ref{alg:csr}) that use the set-valued \\
        &  renegotiation function~$\renegi$ (6)\\
        $\defprogi$ & default program for a program~$\progi$ in $\onefsspacei$ or $ \fullfsspacei$ (3)\\
        $\renegb$ & \selfunction{} function (maps a set of action profiles to \\ &an action profile that is efficient within that set) (5) \\
        % $\intsect$ & agreement set (as a function of a focal player's \\
        % & set-valued renegotiation function, the profile of \\
        % &  their counterparts' set-valued renegotiation \\
        % & functions, and
        % a miscoordinated action profile) \\
        $\pmm $ & Pareto meet minimum (5) \\
        $\yudproji$ & Pareto meet projection for player~$i$ (maps an  \\
        & action profile to a particular set of Pareto-improved
         \\
        & action profiles) (5) \\ \bottomrule
	\end{tabular}
\end{table}

\section{Discussion}
\label{sec:disc}
Using renegotiation to construct 
SPIs in program games
is a rich and novel
area, 
with many directions to explore. 
To name a few:
\begin{itemize}
    \item Which plausible conditions would violate our assumptions about
    players' beliefs used for the PMM guarantee?
    \item What do \textit{unilateral} 
    SPIs \citep{oesterheld2021safe} look like in this setting?
    \item When are 
    SPIs used in sequential, rather than 
    simultaneous-move, settings? 
    In particular, in sequential settings, 
    the first-moving player's decision whether to use a renegotiation program 
    could
    signal private information to the second-moving player.
    \item We have assumed complete information about payoffs;
    using ideas from \citet{DiGiovanni2022Apr}'s framework 
    for program games in the presence of private information, it should 
    also 
    be possible to construct SPIs in incomplete information settings.
    \item How can 
    this theory inform real-world AI system design?
\end{itemize}

\begin{acks}
    For valuable comments and discussions, we thank our anonymous referees, Caspar Oesterheld, Lukas Finnveden, Vojtech Kovarik, and Alex Kastner.
\end{acks}

\balance

\bibliography{main}

\newpage
\onecolumn

\appendix

\section{${n}$-Player Notation and Proof of Proposition \ref{prop:simultaneous}}

\setcounter{theorem}{1}

We extend the 2-player formalism in the main text to the ${n}$-player case as follows
(all other extensions from the 2-player case to the ${n}$-player case are straightforward):
\begin{itemize}
    \item In a program game with program space $\allprog$, each player~$i$'s beliefs~$\priorij$ are a distribution supported on $\bigtimes_{j \neq i} \progspacej$.
    \item A profile of programs, renegotiation functions, etc. with the subscript $-i$ denotes the profile with the $i$th entry excluded. E.g., $\progmi = (\progj)_{j \neq i}$. 
    \item $\yudproji(\fullact)$
    is the set of action profiles~$\genfullact$
    such that
    $\payi(\genfullact) \geq \max\{\pmmi, \payi(\fullact)\}$
    and
    $\payj(\genfullact) = \max\{\pmmj, \payj(\fullact)\}$ for all $j \neq i$. 
    \item For a profile $\fullprog \in \bigtimes_{j=1}^{n} \fsspacej(\renegj)$, let $\renegmi = (\renegj)_{j \neq i}$.
\end{itemize}

The statement of Assumption~\ref{assum:rennopun} in the general ${n}$-player case is:

\setcounter{definition}{3}

\begin{assumption}
We say that players with beliefs~$\fullprior$
% \textbf{are (i) certain 
% that renegotiation won't be punished and (ii) not certain of coordination}
\textbf{are certain 
that renegotiation won't be punished}
if the following holds. For any program profile~$\fullprog$, define~$\defprogmji$ as~$\progmj$ 
    with~$\progi$ replaced by~$\defprogi$.
    Take any renegotiation function $\smreneg \in \fullrenegspace$;
    any renegotiation program
    $\progi \in \onefsspacei(\smreneg)$;
    and any~$\progmi$ in the support of~$\priorij$
    such that the programs in $\fullprog$ don't renegotiate with each other.
    (I.e.,
    there is no $(\smrenegj)_{j \neq i}$
    such that for all $j \neq i$
    we have
    $\progj \in \onefsspacej(\smrenegj)$
    where $\smrenegj(\fullact(\deffullprog)) = \smreneg(\fullact(\deffullprog))$.)
    % (I.e.,~$\progi$
    % is a renegotiation program but
    % doesn't renegotiate with~$\progmi$.)
    % or is played in a subjective equilibrium
    % of~$\game(\allprog)$.
    Then:
    \begin{enumerate}
        \item For all~$j \neq i$, we have $\progj(\progmj) = \progj(\defprogmji)$.
        \item If~$\defprogi$ is used in subjective equilibrium with respect
        to~$\priorij$,
        and
        $\progmi \in \bigtimes_{j \neq i} \onefsspacej(\smrenegj)$ for some~$(\smrenegj)_{j \neq i}$,
        we have~$\defprogi(\progmi) = \defprogi(\defprogmi)$.
    \end{enumerate}
    %\label{assum:pun}
%     \item Let $\notfsspacei$ be the set of $\progi$
%     not in~$\onefsspacei$ (i.e., non-renegotiating
%     programs). For each~$i$,
%     there are some~$(\smrenegj)_{j \neq i}$ and~$\progmi  \in \bigtimes_{j\neq i} \onefsspacej(\smrenegj)$ in the support
%     of~$\priorij$
%     where:\
%      For any 
%      % subjectively optimal non-renegotiation program
%      $\genpi \in \argmax_{\progi \in \notfsspacei} \mathbb{E}_{\progmi \sim \priorij} \progpayi(\fullprog)$,
%     letting~$\deffullact = \fullact((\genpi, \defprogmi))$,
%      we have that~$\smrenegj(\deffullact)$ are equal
%      for all~$j$ and~$\payi(\smreneg(\deffullact)) > \payi(\deffullact)$.
%     % For~$\eqprogi \notin \onefsspacei(\smrenegi)$
%     % for any~$\smrenegi$, in subjective equilibrium of
%     % $\game(\allprog)$,
%     % letting~$\deffullact = \fullact((\eqprogi, \defprogj))$
%     %  we have~$\payi(\reni(\deffullact)) > \payi(\deffullact)$.
%     \label{assum:better}
%     % and assume this property holds for
%     % any~$\progj$ that is
%     % played in a subjective equilibrium
%     % of~$\game(\allprog)$.
% \end{enumerate}
\end{assumption}

In this context, let $\allprogrn = \bigtimes_{i=1}^{{n}} \onefsspacei$. Then:

\spirational*

\begin{proof}
% Let $\notfsspacei$ be the set of $\progi$
% not in~$\onefsspacei$ (i.e., non-renegotiation
% programs).
Let~$(\fulleq, \fullprior)$
be a subjective equilibrium
of $\game(\allprog \cup \allprogrn)$,
where for some~$i$,
$\eqprogi \notin \onefsspacei$;
and
let~$\smreneg$
be an arbitrary
renegotiation function.
% from a 
% profile~$(\smrenegj)_{j \neq i}$
% satisfying
% Assumption~\ref{assum:rennopun}
%\ref{assum:better}.
Let $\progmi$ be in the support of~$\priorij$.
We will
show
% show,
% for a contradiction,
that
against any such $\progmi$,
the
new program $\irfsprogi(\eqprogi) \in \onefsspacei(\smreneg)$
is
weakly
%strictly
better for player~$i$
than~$\eqprogi$
(hence $\irfsprogi(\eqprogi)$ is better in expectation).

% For~$\progj \in \progmi$,
% let~$\progmj^{\irfsprogi(\eqprogi)} = (\progone, \dots, \irfsprogi(\eqprogi), \dots, \prog_{j-1},\prog_{j+1},\dots,\progn)$
% and~$\progmj^{\eqprogi} = (\progone, \dots, \eqprogi, \dots, \prog_{j-1},\prog_{j+1},\dots,\progn)$. 
\begin{itemize}
    \item Suppose there is some~$j$ such
    that~$\progj \notin \onefsspacej$.
    % and $\smrenegsta(\smrenegja, \deffullact) \cap \smrenegja(\smrenegsta, \deffullact) = \emptyset$.
    Then,
    the SPI-transformed program
    doesn't renegotiate,
    so $\irfsprogi(\eqprogi)(\progmi) = \eqprogi(\progmi)$,
    and hence by Assumption~\ref{assum:rennopun}
    %\ref{assum:pun}
    we have, for all $j \neq i$,
    $\progj(\progmj) = \progj(\defeqprogmji)$.
    (Where~$\defprog{\irfsprogi(\eqprogi)} = \eqprogi$.)
    So $\progpayi((\irfsprogi(\eqprogi), \progmi)) = \progpayi((\eqprogi, \progmi))$.
    \item Otherwise, $\progmi \in \bigtimes_{j \neq i} \onefsspacej(\smrenegj)$
    for some renegotiation functions $\smrenegj$. Let $\deffullact = \fullact((\eqprogi, \defprogmi))$. 
    \begin{itemize}
        \item 
         If $\smreneg(\deffullact) \neq \smrenegj(\deffullact)$ for some~$j$,
         then none of the programs renegotiate,
         so 
         $\irfsprogi(\eqprogi)(\progmi) = \eqprogi(\defprogmi)$
         and for all~$j \neq i$ we have $\progj(\progmj) = \defprogj(\defprogmj)$.
         By Assumption~\ref{assum:rennopun},
         %\ref{assum:pun},
         since~$\eqprogi = \defprog{\irfsprogi(\eqprogi)}$ is played in a subjective equilibrium,
         $\eqprogi(\progmi) = \eqprogi(\defprogmi)$
         (while~$\progj(\eqprogi) = \defprogj(\eqprogi)$ for all~$j \neq i$).
         Therefore 
         $\progpayi((\irfsprogi(\eqprogi), \progmi)) = \progpayi((\eqprogi, \defprogmi)) = \progpayi((\eqprogi, \progmi))$.
         % so as above we have $\progpayi((\irfsprogi(\eqprogi), \progj)) = \progpayi((\eqprogi, \progj))$.
        \item If $\smreneg(\deffullact) = \smrenegj(\deffullact)$ for all~$j \neq i$,
        then because~$\smreneg$ and all $\smrenegj$
        are renegotiation functions,
        $(\irfsprogi(\eqprogi), \progmi)$ Pareto-improves on $(\eqprogi, \defprogmi)$,
        so by the above we have 
        $\progpayi((\irfsprogi(\eqprogi), \progmi)) \geq \progpayi((\eqprogi, \progmi))$.
        % And by Assumption~\ref{assum:rennopun}\ref{assum:better},
        % with positive probability
        % we have that~$\progpayi((\irfsprogi(\eqprogi), \progmi)) = \payi(\smreneg(\deffullact)) > \payi(\deffullact) = \progpayi((\eqprogi, \progmi))$.
    \end{itemize}
\end{itemize}
Thus, for all~$\progmi$ we have
$\progpayi((\irfsprogi(\eqprogi), \progmi)) \geq \progpayi((\eqprogi, \progmi))$.

Now, let $\fullprog'$ be the program profile such that:
\begin{itemize}
    \item $\progj' = \eqprogj$ if $\eqprogj \in \onefsspacej(\smrenegj)$ for some $\smrenegj$
    \item $\progj' = \irfsprogj(\eqprogj)$ otherwise.
\end{itemize}

Clearly, then, (1) holds by construction, and each $\progj' \in \onefsspacej$. 
And since the above argument holds for any~$j$,
we have
$$\irfsprogj(\eqprogj) \in \argmax_{\progj \in \progspacej \cup \onefsspacej}
\mathbb{E}_{\progmj \sim \priorji} \progpayj(\fullprog),$$
i.e., (2) $(\fullprog', \fullprior)$ is a subjective equilibrium.
% and with positive probability $\progpayi((\irfsprogi(\eqprogi), \progmi)) > \progpayi((\eqprogi, \progmi))$.
% So~$\eqprogi$ cannot be in a subjective
% equilibrium.
\end{proof}

\section{Proof of Claim that Coordination on the Selection Function Does Not Change Players' Payoffs}
\label{app:barg}

Let $\csprog((\renegi)\indices, \renegb, \deffullprog)$ denote the program profile such that each player~$i$'s program is given by Algorithm~\ref{alg:csr}
for the set-valued renegotiation function $\renegi$
%set~$\fsspaceireni$
and default program $\defprogi$, given the \selfunction{} function~$\renegb$.
Let~$\mathcal{S}$ be the set of \selfunction{} functions.
We claim there exists a mapping~$\genfunc : \bigtimes\indices \svrenegspacei \times \mathcal{S} \to \bigtimes\indices \svrenegspacei$,
returning new set-valued renegotiation functions,
such that the same outcome is induced by the new set-valued renegotiation functions and \textit{any} selection function
as is induced by the old set-valued renegotiation functions and old selection function.
The intuition is that
to the extent
there is a bargaining problem
over selection
functions,
this can be ``translated'' into the players'
choice of 
renegotiation functions.

Formally:\ For any~$\renegb, \renegb' \in \mathcal{S}$ and any~$(\renegi)\indices \in \bigtimes\indices \svrenegspacei$, we have
$\fullact(\csprog(\genfunc((\renegi)\indices, \renegb), \renegb', \deffullprog)) = \fullact(\csprog((\renegi)\indices, \renegb, \deffullprog))$.

To see this, let~$\deffullact = \fullact(\deffullprog)$ and for each~$i$, let~${\renegi}'({\renegmi}', \deffullact) = \{\renegb(\bigcap_{j=1}^n \renegj(\renegmj, \deffullact))\}$.
And let~$\genfunc((\renegi)\indices, \renegb) = ({\renegi}')\indices$.
Then, since $\bigcap_{j=1}^n {\renegj}'({\renegmj}', \deffullact)$
is a singleton, for any~$\renegb'$ we have
$\renegb'(\bigcap_{j=1}^n {\renegj}'({\renegmj}', \deffullact)) = \renegb(\bigcap_{j=1}^n {\renegj}({\renegmj}, \deffullact))$.
So, as required,
$\fullact(\csprog(\genfunc((\renegi)\indices, \renegb), \renegb', \deffullprog)) = \fullact(\csprog((\renegi)\indices, \renegb, \deffullprog))$.

The reason this result implies
that players do not face a
bargaining problem
when coordinating on
the selection function is as follows.
Consider a variant of a program
game in which:
\begin{itemize}
    \item Instead of choosing just one CSR
    program for a specific selection function,
    each player
    independently
    chooses a ``meta''-CSR program --- which,
    for each possible selection function,
    specifies a CSR
    program (using the same default program)
    for that selection function.
    \item At the same time as the meta-CSR programs are submitted,
    one player~$i$ (chosen arbitrarily)
    chooses a selection function~$\renegb_i$.
    \item Each other player~$j$, in sequence, reports a selection function~$\renegb_j$.
    \item If the players agree on the same selection function,
    the players' corresponding
    CSR programs are played. Otherwise, their default programs are used.
\end{itemize}
Recall
that the claim we proved above
is:\ all possible outcomes
of CSR program profiles can be attained under
all possible selection functions,
via players varying their renegotiation functions
across the programs in their
meta-CSR program.
Given this claim,
it is reasonable to assume that players 
expect the same distribution
of outcomes
for each possible 
agreed-upon
selection function.
So the players $j \neq i$
have no reason not to accept the selection function
chosen by player~$i$.

% Consider a profile of CSR programs $\fullprog \in \bigtimes\indices \fsspaceireni$, for any \selfunction{} function~$\renegb$.

% % We want to show that $\progi$ plays equivalently to some~$\progi' \in \fsspacei({\renegi}')$ for some 
% % common \selfunction{} function~$\renegb$, 
% % i.e., for any~$\progmi \in \bigtimes_{j \neq i}  \modspacei(\renegj,\renegbj)$ for some tuple of \selfunction{} functions~$(\renegbj)_{j \neq i}$,
% % we have~$\progi(\progmi) = \progi'(\progmi')$ \dots
% % % a player~$i$ with a program~$\progi \in \fsspaceireni$.
% % % We want to show that, for any~$\progmi \in \bigtimes_{j \neq i} \fsspacej(\renegj)$,

\section{Proof of Theorem~\ref{prop:yudpoint} }

Denote 
the \textbf{agreement set} as 
     $\intsect(\renegi, \renegmi, \fullact) = \bigcap_{j=1}^{n} \renegj(\renegmj, \fullact)$.

In the ${n}$-player setting, define:

\setcounter{definition}{6}
     \begin{definition}
For any~$\progi \in \fsspaceireni$, the \textbf{PMP-extension}~$\widetilde{\progi} \in \fsspaceinewreni$
    is the program
    identical to~$\progi$
    except:\
    for all $\progmi$,
    writing~$\newfullprogi = (\newprogi, \progmi)$,
    we have
    \begin{align*}
        \newrenegi(\renegmi, \fullact(\defnewfullprogi)) &= \renegi(\renegmi, \fullact(\defnewfullprogi)) \cup \yudproji(\fullact(\fullprog)).
    \end{align*}
    Let~$\newrenegimj$ denote~$\renegmj$ with~$\renegi$ replaced by~$\newrenegi$. 
\end{definition}

\begin{assumption}
     We say that 
     players 
     with beliefs $\fullprior$ 
     \textbf{are (i) certain that CSR won't be punished and (ii) certain that PMP-extension won't be punished} 
     if the following hold:
     \begin{enumerate}[(i)]
     %     \item Suppose either $\defprogi$ is in a subjective equilibrium
     % of~$\game(\allprog \cup \bigtimes_{i=1}^{{n}} \fullfsspacei)$,
     % or, for some player $j$, $\defprogi$ is an element of some $\progmj$ in the support of $\priorji$.
     % Let $\progi \in \fsspaceireni$ with $\renegi$ such that
     %     $\renegi({\renegmi}', \fullact) = \emptyset$ for all ${\renegmi}'
     %     , \fullact$,
     %     and default program $\defprogi$.
     %     Then 
     %     we have 
     %     $\defprogi(\defprog{\progmi}) = \defprogi(\progmi)$.
      \item Suppose either $\progi$ is in a subjective equilibrium
     of~$\game(\allprog \cup \allprogreneg)$,
     or, for some player $j$, $\progi$ is an element of some $\progmj$ in the support of $\priorji$.
      Suppose $\progi \notin \fsspacei$. 
         Then for any $\progmi \in \bigtimes_{j \neq i} \fsspacej$,
         we have 
         $\progi(\defprogmi) = \progi(\progmi)$.
         \label{subassum:csrwlog}
         \item Let~$\progmi \in \bigtimes_{j \neq i} \fsspacej(\renegj)$ be in the
     support of~$\priorij$,
     and take any~$\progi \in \fsspaceireni$
     with PMP-extension $\newprogi$.
         For all~$j \neq i$ and all~$\fullact$, we have
     $\renegj(\newrenegimj,\fullact) = \renegj(\renegmj,\fullact) \cup 
     {V}$ for some ${V} \subseteq \yudproji(\fullact(\fullprog))$. \label{subassum:extension}
     \end{enumerate}
     % (with corresponding PMP-extension~$\newprogi$).
     %we have
     % \begin{align*}
     %     \renegj(\renegmj,\fullact) \subseteq \renegj(\newrenegimj,\fullact) \\
     %  \subseteq \renegj(\renegmj,\fullact) \cup 
     % \yudproji(\bargdefault(\renegi, \renegmi, \fullact)) 
     % \end{align*}
     % $\renegj(\newrenegi,\fullact) = \renegj(\renegi,\fullact)$
     
     % and for~$\progmi \notin \notsafe(\progi)$
     % we have
     % $$\renegj(\newrenegimj,\fullact(\defnewfullprogi)) = \renegj(\renegmj,\fullact(\defnewfullprogi)).
     % $$
      \label{assum:nopunish}
\end{assumption}

\pmmguarantee*

\begin{proof}
Let~$(\fullprog, \fullprior)$
be a subjective equilibrium
of $\game(\allprog \cup \bigtimes_{i=1}^{{n}} \fullfsspacei)$.
First, we can assume each $\progi$ is in $\fsspacei$.
To see this, take any $\progi' \notin \fsspacei$,
and define a corresponding
CSR program $\progi'' \in \fsspacei({\renegi}'')$ by:
\begin{itemize}
    \item $\defprog{\progi''} = \progi'$, and
    \item For all ${\renegmi}, \fullact$, we have ${\renegi}''({\renegmi}, \fullact) = \emptyset$.
\end{itemize}
Then, consider any $\progmi$. If $\progmi \notin \bigtimes_{j \neq i} \fsspacej(\renegj)$ for any $(\renegj)_{j\neq i} \in \bigtimes_{j \neq i} \fullsvspace$, then
the renegotiation procedure doesn't
occur,
i.e., $\progi''(\progmi) = \progi'(\progmi)$.
Otherwise, since $\progi''$ always 
returns an empty renegotiation set,
$\progi''(\progmi) = \progi'(\defprog{\progmi})$. But by
Assumption~\ref{assum:nopunish}\ref{subassum:csrwlog},
we have $\progi'(\defprog{\progmi}) = \progi'(\progmi)$.
Thus $\progi''$ has the same outputs as $\progi'$, for all input programs, as required.

Given this, 
let
% Without loss of generality,
$\progi \in \fsspaceireni$
for some $\renegi$.
% Let~$\notsafe(\progi)$
% be the set of~$\progj \in \fsspacej(\renegj)$
% for some~$\renegj$
%     such that $\payip(\bargdefault(\fullprog, \fullact(\deffullprog))) < \pmmip$ for some~$i'$.
Define~$\newprogi \in \fsspaceinewreni$
by:
\begin{itemize}
    \item 
    $\defprog{\newprogi} = \defprogi$, and
    % \item
    % $\newrenegki \equiv \renegki$
    % for ${k} \geq 2$, and
    \item
    For all $\progmi$,
    we have
    $$\newrenegi(\renegmi, \fullact(\defnewfullprogi)) = \renegi(\renegmi, \fullact(\defnewfullprogi)) \cup 
    \yudproji(\fullact(\fullprog)).
    $$
\end{itemize}
We will show
% for a contradiction,
that this new program
is at least as subjectively good for player~$i$
as the original $\progi$.
Consider a~$\progmi$ in the support of~$\priorij$. By the same argument as above,
we can assume $\progmi \in \bigtimes_{j \neq i} \fsspacej(\renegj)$,
for some set-valued renegotiation functions~$(\renegj)_{j \neq i}$:

\begin{itemize}
        \item \textbf{PMP of default renegotiation outcome
        not in a counterpart's set:}
         By construction
    $\fullact(\defnewfullprogi) = \fullact(\deffullprog)$.
        If for some $j \neq i$, player~$i$'s PMP
    of the
    default
    renegotiation outcome
    isn't in~$j$'s
    renegotiation set
    (in response to
    player~$i$'s modified renegotiation set),
    then 
    adding that projection
    to~$i$'s
    set doesn't make a difference.
    That is,
    let
    ${S}^{\yudproji} = \yudproji(\fullact(\fullprog)) \cap \bigcap_{j \neq i} \renegj(\newrenegimj, \fullact(\deffullprog))$.
    If~${S}^{\yudproji} = \emptyset$,
    then by Assumption~\ref{assum:nopunish}\ref{subassum:extension}, $\intsect(\newrenegi, \renegmi, \fullact(\defnewfullprogi)) = \intsect(\renegi, \renegmi, \fullact(\deffullprog))$,
    % $\progj$
    % also responds identically
    % to these programs,
    so~$\fullprogpay(\newfullprogi) = \fullprogpay(\fullprog)$.
        \item 
        If~${S}^{\yudproji} \neq \emptyset$:
            \begin{itemize}
            \item
            By definition, ${S}^{\yudproji} \subseteq \yudproji(\fullact(\fullprog))$.
            So, by Assumption~\ref{assum:nopunish}\ref{subassum:extension},\footnote{Remark:\ Notice that this could be relaxed to: ``Let~$\progmi \in \bigtimes_{j \neq i} \fsspacej(\renegj)$ be in the
     support of~$\priorij$,
     and let~$\progi \in \fsspaceireni$
     be in a subjective equilibrium
     of~$\game(\allprog \cup \bigtimes_{i=1}^{{n}} \fullfsspacei)$
     with PMP-extension $\newprogi$.
         For all~$j \neq i$ and all~$\fullact$, we have
     $\renegj(\newrenegimj,\fullact) = \renegj(\renegmj,\fullact) \cup 
     {V}$ for some ${V} \subseteq \yudproji(\fullact(\fullprog))$.'' We state the technically stronger assumption in the main text for simplicity.}
            the only change to the agreement
            set due to the PMP-extension
            is that outcomes from
            player~$i$'s PMP of the default
            renegotiation outcome
            are added:\
            $\intsect(\newrenegi, \renegmi, \fullact(\deffullprog)) = \intsect(\renegi, \renegmi, \fullact(\deffullprog)) \cup {S}^{\yudproji}$,
            where this set is nonempty.
            (Thus the new renegotiation outcome
            is~$\fullact(\newfullprogi) = \renegb(\intsect(\newrenegi, \renegmi, \fullact(\deffullprog))) $.)

           \item Lastly, we consider the new
            renegotiation outcome
            given the two cases for 
            the default
            renegotiation outcome
            $\fullact(\fullprog)$:
            \begin{enumerate}
                 \item \textbf{Agreement
                 achieved without PMP:} Suppose~$i$'s original program
                 reached agreement
                 with the other players,
                 that is, we have $\fullact(\fullprog) = \renegb(\intsect(\renegi, \renegmi, \fullact(\deffullprog)))$.
                For all $\textbf{x} \in \yudproji(\fullact(\fullprog))$, the following holds by the
                definition
                of PMP:\ $\fullpay(\textbf{x}) \succeq \fullpay(\fullact(\fullprog))$.
                % either all outcomes
                % in~$\yudproji(\bargdefault(\fullprog, \fullact(\deffullprog)))$
                % %~$\yudproj(\renegb(\intsect(\renegki, \renegkj, \deffullact^{(k)}(\fullprog))))$ 
                % Pareto-improve
                % on~$\bargdefault(\fullprog, \fullact(\deffullprog))$
                % or~$\bargdefault(\fullprog, \fullact(\deffullprog)) \in \yudproji(\bargdefault(\fullprog, \fullact(\deffullprog)))$
                % already,
                Since~$\renegb$ is a transitive
                \selfunction{} function,
                therefore,
                $$\fullpay(\fullact(\newfullprogi)) \succeq 
                \fullpay(\fullact(\fullprog)).
                $$
                % $\bargdefault(\newfullprogi, \fullact(\deffullprog)) \in \yudproji(\bargdefault(\fullprog, \fullact(\deffullprog)))$
                % and so it Pareto-improves on~$\bargdefault(\fullprog, \fullact(\deffullprog))$.
                \item \textbf{No agreement without PMP:} Otherwise, since by construction we have $\fullact(\newfullprogi) \in \yudproji(\fullact(\fullprog))$,
                then
                $$
                \fullpay(\fullact(\newfullprogi)) \succeq 
                \fullpay(\fullact(\fullprog)).
                $$
                % Pareto-improves on~$\bargdefault(\fullprog, \fullact(\deffullprog))$.
            \end{enumerate}
            \item Thus in either
            of the two cases,
            $\progpayi(\newfullprogi) \geq \progpayi(\fullprog)$.
        \end{itemize}
    \end{itemize}

% where by assumption~$\progj \in \fsspacej(\renegj)$
% for some~$\renegj$.

% meaning~$\fullprogpay(\newfullprog) = \fullprogpay(\fullprog)$.

Thus, for all~$\progmi$ we have
$\progpayi(\newfullprogi) \geq \progpayi(\fullprog)$.
Applying the same
argument for each player~$j$,
and with~$\newfullprog = (\newprogi)_{i=1}^{{n}}$,
it follows that~$(\newfullprog, \fullprior)$
is a subjective equilibrium.
Now, assume this profile does \textit{not} guarantee the PMM,
that is, 
for some~$i'$, 
$\payip(\fullact(\newfullprog)) < \pmmip$.
Let~$\newrenegmi = (\newrenegj)_{j \neq i}$.
Since~$\yudproji(\fullact(\newfullprog)) \subseteq \newrenegi(\newrenegmi, \fullact(\defprog{\newfullprog}))$
for all~$i$,
it follows from the argument above that the players'
renegotiation sets have nonempty
intersection,
and that~$\renegb(\intsect(\newrenegi, \newrenegmi, \fullact(\defprog{\newfullprog}))) \in \bigcap_{i=1}^{{n}} \yudproji(\fullact(\newfullprog))$ guarantees
each player at least~$\pmmi$.
This contradicts the assumption that
$\newfullprog$ doesn't guarantee
the PMM,
% $\newfullprogmi \in \notsafe(\newprogi)$,
so $\fullprogpay(\newfullprog) \succeq \pmm$ as required.
\end{proof}

\section{Iterated CSR and Tightness of the PMM Bound}
\label{app:proof:ineff}

In the main
text,
we considered renegotiation that takes place in one round.
We might expect, however, that
if 
players
renegotiate
for \textit{multiple}
rounds
indefinitely,
and they are required
to take a strict Pareto
improvement
at every round of renegotiation,
they are guaranteed
payoffs that nontrivially
exceed
the PMM.
As we will
show,
this is not always true.

Consider \termemph{iterated
CSR (ICSR)} programs, constructed as follows.
An ICSR program
works
by repeating the procedure
executed by a CSR program
for~${K}$ rounds,
using the renegotiation outcome
of the previous round
as the default outcome for the next
round.
Formally:\
% for some number of rounds~${K}$,
Consider a tuple~$\fullrenegi$ $= ({\reneg}^{i,(k)})$ $_{k=1}^{{K}} \in \bigtimes_{k=1}^{{K}} \fullsvspace$.
% a
% \textit{renegotiation strategy}.
For such a tuple~$\fullrenegi$,
% for a given number of rounds~${K}$,
we define
the space of
% iterated renegotiation
ICSR
programs
$\itspaceireni$
as the
space of programs
with the structure
of
Algorithm~\ref{alg:icsr},
for some default~$\defprogi$.
(For each $i$, 
let~$\fullitspacei = \bigcup_{(\fullrenegi)} \itspaceireni$,
i.e.,
the space of all
%iterated renegotiation
ICSR
programs.)
% define \textit{iterated}
% CSR programs~$\progi \in \itspaceireni$ as follows.
% For~${k} = 1, \dots, {K}$,
% there are~$\progi^{(k)} \in \fsspacei(\renegki)$,
% such that:
% \begin{itemize}
%     \item For~$\progj \in \itspacei(\fullrenegj, {K})$,
%     we have~$\fullact((\defprog{\progi^{(k)}}, \defprog{\progj^{(k)}})) = \fullact((\progi^{(k-1)}, \progj^{(k-1)}))$ for~${k} = 2, \dots, {K}$.
%     % \item Define~$\deffullact^{(1)}(\fullprog) = \fullact((\defprog{\progi^{(1)}}, \defprog{\progj^{(1)}}))$
%     % and for~${k} = 2, \dots, {K}$,
%     % let
%     % $\deffullact^{(k)}(\fullprog) = \fullact((\progi^{(k-1)},\progj^{(k-1)}))$.
%     \item If~$\progj \notin \itspacei(\fullrenegj, {K})$ for any~$\fullrenegj$,
%     we have~$\progi(\progj) = \defprog{\progi^{(1)}}(\progj)$.
%     \item Otherwise,~$\progi(\progj) = \fullact((\progi^{(K)},\progj^{(K)}))$.
% \end{itemize}
Let the outcome of the $k$th round of
renegotiation from point~$\fullact$ using ICSR programs~$\fullprog$, if any, be 
\begin{align*}
\bargdefault^{(k)}(\fullprog, \fullact) &= 
     \begin{cases}
         \renegb(\intsect(\renegki, \renegkim, \fullact)),& \text{if }\intsect(\renegki, \renegkim, \fullact) \neq \emptyset, \\
         \fullact,& \text{else.}
    \end{cases}
\end{align*}

\setcounter{theorem}{4}

Then, Proposition~\ref{proposition:inefficiency}
shows that
our PMM payoff bound is tight
in bargaining problems,
that is, games
where no more than one player can achieve
their best feasible payoff.
The intuition for this result is given
in Example~\ref{ex:iterated}.
% Each player may be confident that
% their counterpart will 
% make increasingly favorable offers at each round of
% renegotiation, such that 
% they make arbitrarily small concessions at each round. Two such 
% players will thus arrive at an arbitrarily small improvement on the
% PMM.

% \begin{sketch}
% Consider a player~$i$ who is certain
% that, by the end of renegotiation,
% player~$j$ will concede~$i$'s most-preferred
% outcome.
% Then~$i$'s best response is
% to make a sequence of smaller
% concessions
% at each round,
% such that the total concession to~$j$ (added to~$j$'s PMM)
% is~$\delta$
% --- which we can choose
% to be arbitrarily small.
% % such that the resulting
% % subjective equilibrium
% % is inefficient.
% \end{sketch}

\begin{algorithm}
\caption{
Iterated
conditional
set-valued
renegotiation
% ${K}$-Iterated
% renegotiation
%Failsafe-implementing 
program $\progi \in \itspaceireni$, for some $\defprogi$}
\label{alg:icsr}
\begin{algorithmic}[1]
\Require Counterpart program profile $\progmi$
\If{
$\progmi \in \bigtimes_{j \neq i} \itspacej(\fullrenegj)$ for some $(\fullrenegj)$ $_{j\neq i} \in \bigtimes_{j \neq i}  \bigtimes_{{k}=1}^{K} \fullsvspace$}
    \If{${k} = 1$}
            \State 
             $\deffullact \leftarrow \fullact(\deffullprog)$ %$\deffullact \leftarrow (\defprog{\progone}(\defprog{\progtwo}), \defprog{\progtwo}(\defprog{\progone}))$
        % \Else
        %     \State $\deffullact \leftarrow \dots $
        \EndIf
    \For{${k} \in \{1,\dots,{K}\}$} \Comment{Renegotiation rounds}
        % \State $(\acceptone, \accepttwo) \leftarrow (\fullrenegkone(\fullrenegktwo, \deffullact), \renegktwo(\renegkone, \deffullact))$
        % \State ${I} \leftarrow \acceptone \cap \accepttwo$
        \State ${I} \leftarrow \bigcap_{j = 1}^{{n}} \renegkj(\renegkjm, \deffullact)$
       % \State ${I} \leftarrow \renegkone(\renegktwo, \deffullact) \cap \renegktwo(\renegkone, \deffullact)$
        \Comment{Agreement set}
         \If{${I} \neq \emptyset$}
    % and $\renegb({I}) = \renegjb({I}) $}
    % \Comment{Check agreement on first round}
    \label{algline:firstagree}
             \State $\deffullact \leftarrow \renegb({I})$%\If{$(\progpayone((\defprog{\progone},\defprog{\progtwo})), \progpaytwo((\defprog{\progone},\defprog{\progtwo})))$ is efficient}
             \Comment{Update renegotiation outcome}
    %     \State \Return $\defacti$  \Comment{Play default}
    %     \label{algline:default}
    \EndIf
    \EndFor
    \State \Return $\defacti$
    \Comment{Play final renegotiation outcome, or default}
\Else
     \State \Return $\defprogi(\progmi)$  % \Comment{Play default}
\EndIf
\end{algorithmic}
\end{algorithm}

\begin{exmp}
\label{ex:iterated}
\textbf{(Failure to improve significantly 
on the PMM despite iterated renegotiation.)} 
Consider a two-player game.
Suppose that each player~$i$
believes the other will
in round~${k}$
use an
unconditional
renegotiation set of the form, 
``Accept any Pareto improvements
that give both of us at least our PMM payoff, 
and gives my counterparts at most~$\payi^{(k)}$,''
for some increasing sequence of upper bounds~$\{\payi^{(k)}\}_{{k}=1}^{{K}}$.
Intuitively, each upper bound is an ``offer'' of some amount on
the Pareto frontier.
In particular,~$i$ believes
that most likely~$j$'s set in the final renegotiation
round will accept any outcomes
that leave~$i$ just slightly worse off
than in~$i$'s most-preferred outcome;
otherwise, they will accept all outcomes.

Then, if player~$i$
must make some strictly Pareto-improving offer each round,
their
best response
is to make an offer each round
small enough
that at the end of renegotiation,~$j$
will offer slightly less than~$i$'s
most-preferred outcome.
But, if player~$j$ has the same beliefs about~$i$,
the players only slightly improve
upon the PMM.
\end{exmp}

\begin{restatable}
{proposition}
{tight}
\label{proposition:inefficiency}
Write $\bestparetopayoffi$ for
player $i$'s best feasible
payoff. Take a 
non-zero-sum game $G$
with a
feasible set that is continuous, 
contains the Pareto meet, and 
such that for every 
feasible $\fullpay$ with 
$\payi = \bestparetopayoffi$ for some~$i$, 
$\payj < \bestparetopayoffj$ for each 
$j \neq i$. 
% $(\bestparetopayoffi)_{i=1}^n$.
% $(\bestparetopayoffone, 
% \bestparetopayofftwo)$.
Let~$\game(\allprog)$
be a program
game,
and~${K}$
be any natural number.
% For both $i$, 
% let~$\fullitspacei = \bigcup_{(\renegi, {K})} \fsspacei(\renegi, {K})$.
Then for all $\Delta \succ \mathbf{0}$ 
there exists a subjective equilibrium~$(\fulleq, \fullprior)$
of~$\game(\allprog \cup \bigtimes_{i=1}^{{n}} \fullitspacei)$
(where each~$\eqprogi \in \itspaceireni$
for some~$\fullrenegi$) satisfying the assumptions of Theorem \ref{prop:yudpoint}
in which 
\begin{enumerate}
    \item
    Let~$\deffullact^{(1)} = \fullact(\defprog{\fulleq})$,
    and for~${k} \in \kmlist$,
    % let~$\deffullact^{(k+1)} = \renegb(\intsect(\renegki, \renegkj, \deffullact^{(k)}))$.
    let~$\deffullact^{(k+1)} = \bargdefault^{(k)}(\fullprog, \deffullact^{(k)})$.
    Then, for
    % both~$i$ and
    every~${k} \in \kmlist$,
    $\deffullact^{(k+1)}$
    strictly Pareto-improves
    on~$\deffullact^{(k)}$;
    \item $\fullprogpay(\fulleq) \preceq
    \pmm + \Delta$.
\end{enumerate}
\end{restatable}

% \adigi{re-check this}
\begin{proof}
Because $G$ is non-zero-sum 
and has a continuous feasible set, 
we have that $\pmm + (\delta_1, \dots, \delta_n)$ is
feasible for sufficiently small $\delta_0$ 
and $\delta_i \in (0, \delta_0)$ for each $i$.
By the assumptions that
for every 
feasible $\fullpay$ with 
$\payi = \bestparetopayoffi$ for some $i$, 
$\payj < \bestparetopayoffj$ for each 
$j \neq i$, and
that the feasible set is continuous,
we can take $\epsdeltaone > 0$ such that
$(\bestparetopayoffone-\epsdeltaone, \pmmtwo + \delta_2, \dots, \pmmn + \delta_n)$ 
is Pareto efficient.
We will first construct subjective 
beliefs
for player 1
about the other players'
programs, and 
show that the best-response to 
these beliefs results in player~1 including points in their
renegotiation sets that 
give each other player~$j$ no more than 
%$\tilde{u}_2$.
$\pmmj + \delta_j$.
Then, if
each player~$j$ has symmetrical 
beliefs about player 1, 
the resulting subjective 
equilibrium is inefficient.  

\par Fix natural number $K$ and \selfunction{} function 
$\renegb$.
Abusing notation,
let~$\payi(\yudproj(\fullact)) = \min_{\genfullact \in \yudproji(\fullact)} \payi(\genfullact)$. (Notice that $\min_{\genfullact \in \yudproji(\fullact)} \payi(\genfullact) = \min_{\genfullact \in \yudprojj(\fullact)} \payi(\genfullact)$
for all~$j \neq i$.)
%$\payi(\genfullact)$
%for any~$\genfullact \in \yudproj(\fullact)$
%(which all
%induce the same payoff profile).
For $k=1,\dots, K$ and $\fullact$ 
we have for $j \neq 1$, $z \in \{x, y\}$:
\begin{equation*}
   \begin{aligned}
       \reneg^{1,(k)}(\reneg^{-1,(k),z}, \fullact) & =  
       \begin{cases}
       \left\{\fullact': \fullpay(\fullact') \succeq \fullpay(\yudproj(\fullact)), \payj(\fullact') \leq 
       \pmmj + \frac{k}{K}\delta_j \text{ for all $j \neq 1$}\right\}, \\
       \ \ \ \payj(\fullact) \leq 
       \pmmj + \frac{k}{K}\delta_j \text{ for all $j \neq 1$};\\
       \emptyset, \mathrm{ otherwise}.\end{cases};\\
       \reneg^{j,(k),x}(\reneg^{-j,(k)}, \fullact) & =  
       \begin{cases}
       \left\{\fullact': \fullpay(\fullact') \succeq \fullpay(\yudproj(\fullact)), u_1(\fullact') \leq \pmmone + \frac{k}{K}(\bestparetopayoffone-\pmmone) \right\}, \\
       \ \ \ u_1(\fullact) \leq \pmmone + \frac{k}{K}(\bestparetopayoffone - \pmmone);\\ 
       \emptyset, \mathrm{ otherwise}.\end{cases};\\
       \reneg^{j,(k),y}(\reneg^{-j,(k)}, \fullact) & =  
       \begin{cases}
       \left\{\fullact': \fullpay(\fullact') \succeq \fullpay(\yudproj(\fullact)), u_1(\fullact') \leq \pmmone + \frac{k}{K}(\bestparetopayoffone-\epsdeltaone-\pmmone) \right\}, \\
       \ \ \ u_1(\fullact) \leq \pmmone + \frac{k}{K}(\bestparetopayoffone - \epsdeltaone - \pmmone);\\ 
       \emptyset, \mathrm{ otherwise}.\end{cases}\\
   \end{aligned} 
\end{equation*}
Let $\progone$, $\progmone^x$, and $\progmone^y$  
be the 
% iterated renegotiation
CSR
programs defined respectively by these renegotiation functions,
along with some default programs which result 
in default payoffs Pareto-worse than
$\pmm$.
Using $\progone$ player 1 attains a payoff of 
at least $\bestparetopayoffone-\epsdeltaone$ against $\progmone^x$ 
and a payoff of exactly $u^*_1-\epsdeltaone$ against 
$\progmone^y$. Let $\beta=\beta_{1}(\progmone^x)$ and 
$1-\beta=\beta_{1}(\progmone^y)$. Thus player $1$'s 
expected payoff using $p_1$ is at least
$\bestparetopayoffone-\epsdeltaone$. Player 1 cannot 
improve their payoff against $\progmone^x$ or $\progmone^y$ by 
conceding more than $\epsdeltaone$, and conceding
strictly between $0$ and $\epsdeltaone$ will result in a payoff
of strictly less than $\bestparetopayoffone$ against
$\progmone^x$
and at most
$\pmmone + \frac{K-1}{K}(\bestparetopayoffone-\epsdeltaone
-\pmmone)$ against $\progmone^y$. Thus any program that concedes less than 
$\epsdeltaone$ has a payoff bounded above by 
$\beta\bestparetopayoffone + (1-\beta)(\pmmone + \frac{K-1}{K}(\bestparetopayoffone-\epsdeltaone
-\pmmone))$. We can choose $\beta$ small enough to make this smaller than 
$\bestparetopayoffone-\epsdeltaone$, such that~$p_1$ is a best response to beliefs~$\beta_{1}$. 
\par Now, we can construct
symmetric beliefs $\priorji$ for  each
player~$j$, such that a symmetric program $\progj$
is a best response to these beliefs, and $(\fullprog, \fullprior)$
is a subjective equilibrium. And, these programs
played against each other will result in a 
payoff profile Pareto-dominated by 
$\pmm + (\delta_1, \dots, \delta_n)$.
Thus the subjective equilibrium is 
inefficient. This is even though 
players' renegotiation sets 
overlap at each step of renegotiation, and
so their payoffs strictly improve 
at each step.

% At the last round $K$ of renegotiation, we have 
% $\reneg^{1,(K)}(\reneg^{2,(K)}, \deffullact^{(K-1)}) \cap  
% \reneg^{2,(K)}(\reneg^{1,(K)}, \deffullact^{(K-1)}) 
% = \left\{\fullact': \fullpay(\fullact') \succeq \fullpay(\deffullact^{(K-1)}),  
% u_2(\fullact') \leq \tilde{u}_2
% \right\}$. Since the action returned by 
% $(\progone, \progtwo')$
% is given by the \selfunction{} function $\renegb$ applied to this set, 
% and is therefore a Pareto-efficient element of this set, 
% player 1's payoff is
% at least $\bestparetopayoffone-\epsilon$. 
% Now construct a program $p_2''$ identical to 
% $p_2'$, except with renegotiation sets 
% whose best payoff for player 1 increments
% by $\frac{1}{K}(\bestparetopayoffone - \epsilon 
% - \pmmone)$ at each round of renegotiation. 
% For sufficiently large probability on $p_2''$, 
% player 1's best-response to $q_{12}$ is $p_1$.

% \par Now, suppose that player 2 has symmetrical beliefs about player 1, 
% and we construct a symmetrical program $\progtwo$ that is an 
% $\epsilon$-best response to these beliefs, 
% such that $((p_1, p_2), (q_{12}, q_{21}))$
% is a subjective equilibrium.
% And, 
% $\fullprogpay((\progone, \progtwo)) \preceq (u_1^{\mathrm{M}} + \delta, u_2^{\mathrm{M}}+\delta)$, 
% and so $(p_1, p_2)$ is inefficient.

%\par \jc{TODO: update for new construction} 
Checking that this 
%$\epsilon$-
subjective equilibrium
satisfies the assumption of Theorem~\ref{prop:yudpoint}:
Each player~$i$'s beliefs put probability~1
    on the other players using programs
    whose renegotiation sets~$\renegkj$
    are independent of~$\renegki$.
    Thus such a program does not respond differently
    to~$\progi$ and~$\newprogi$
    as defined in Definition~\ref{def:pmp-extension}.
    
\end{proof}

\end{document}